\documentclass[11pt]{article}
\usepackage{fullpage}
\usepackage{amsmath,amsfonts,amsthm,amssymb,url}
\usepackage{url}
\usepackage{color}
\usepackage[usenames,dvipsnames,svgnames,table]{xcolor}
\usepackage[colorlinks=true, linkcolor=red, urlcolor=blue, citecolor=gray]{hyperref}
\usepackage{graphicx}
\usepackage{caption} \usepackage{subcaption, enumitem}
\usepackage{hhline}
\usepackage{algorithm}
\usepackage{algpseudocode}
\usepackage{float}
\usepackage{threeparttable}
\makeatletter

\usepackage{mathtools}

\DeclareMathOperator*{\E}{\mathbb{E}}
\let\Pr\relax
\DeclareMathOperator*{\Pr}{\mathbb{P}}

\DeclareMathOperator{\supp}{supp}
\DeclareMathOperator{\rank}{rank}

\DeclareMathOperator{\orth}{orth}
\newcommand{\R}{\mathbb{R}}

\DeclareMathOperator*{\argmin}{arg\,min}

\newcommand{\poly}{\mathop\mathrm{poly}}
\newcommand{\wh}{\widehat}

\DeclareMathOperator{\nnz}{nnz}
\DeclareMathOperator{\diag}{diag}

\newcommand{\eqdef}{\mathbin{\stackrel{\rm def}{=}}}
\newcommand{\norm}[1]{\|#1\|}
\newcommand{\bnorm}[1]{\left\|#1\right\|}

\newcommand{\todo}[1]{}
\newcommand{\Cam}[1]{}

\usepackage{color}

\makeatletter

\newcommand{\spara}[1]{\smallskip\noindent{\bf #1}\smallskip}

\newtheorem*{rep@theorem}{\rep@title}
\newcommand{\newreptheorem}[2]{%
\newenvironment{rep#1}[1]{%
 \def\rep@title{#2 \ref{##1}}%
 \begin{rep@theorem}}%
 {\end{rep@theorem}}}
\makeatother
\newtheorem{theorem}{Theorem}
\newreptheorem{theorem}{Theorem}
\newtheorem{conjecture}[theorem]{Conjecture}
\newtheorem{corollary}[theorem]{Corollary}
\newtheorem{lemma}[theorem]{Lemma}
\newtheorem*{lemma*}{Lemma}
\newreptheorem{lemma}{Lemma}

\newtheorem{claim}[theorem]{Claim}
\newreptheorem{claim}{Claim}
\newtheorem{definition}[theorem]{Definition}
\newtheorem{problem}{Problem}

  \usepackage{nth}
  \usepackage{intcalc}

\title{Simple Heuristics Yield Provable Algorithms for \protect\\ Masked Low-Rank Approximation}
\date{}
\author{
Cameron Musco\\ UMass Amherst\\ \texttt{cmusco@cs.umass.edu}
\and
Christopher Musco\\ New York University \\ \texttt{cmusco@nyu.edu}
\and 
David P. Woodruff\\ Carnegie Mellon University\\ \texttt{dwoodruf@cs.cmu.edu}
}
\begin{document}
\maketitle

\begin{abstract}
In the \emph{masked low-rank approximation problem}, one is given data matrix $A \in \R^{n \times n}$ and binary mask matrix $W \in \{0,1\}^{n \times n}$. The goal is to find a rank-$k$ matrix $L$ for which: $$\textrm{cost}(L) \eqdef \sum_{i=1}^{n} \sum_{j = 1}^{n} W_{i,j} \cdot (A_{i,j} - L_{i,j} )^2 \leq OPT + \epsilon \|A\|_F^2 ,$$ where $OPT = \min_{\textrm{rank-}k\ \hat{L}} \textrm{cost}(\hat L)$ and $\epsilon$ is a given error parameter. Depending on the choice of $W$, the above problem captures  factor analysis, low-rank plus diagonal decomposition, robust PCA, low-rank matrix completion, low-rank plus block matrix approximation, low-rank recovery from monotone missing data, and a number of other important problems. Many of these problems are NP-hard, and while algorithms with provable guarantees are known in some cases, they either 1) run in time $n^{\Omega(k^2/\epsilon)}$ or 2) make strong assumptions, for example, that $A$ is incoherent or that the entries in $W$ are chosen independently and uniformly at random.

In this work, we show that a common polynomial time heuristic, which simply sets $A$ to $0$ where $W$ is $0$, and then finds a standard low-rank approximation, yields bicriteria approximation guarantees for this problem. In particular, for rank $k' > k$ depending on the \emph{public coin partition number} of $W$, the heuristic outputs rank-$k'$ $L$  with cost$(L) \leq OPT + \epsilon \|A\|_F^2$. 
 This partition number is in turn bounded by  the {randomized communication complexity} of $W$, when interpreted as a two-player communication matrix. For many important cases, including all those listed above, this yields bicriteria approximation guarantees with rank $k' = k \cdot \poly(\log n/\epsilon)$.

Beyond this result, we show that different notions of communication complexity yield bicriteria algorithms for natural variants of masked low-rank approximation. For example, multi-player number-in-hand communication complexity connects to masked tensor decomposition and non-deterministic communication complexity to masked Boolean low-rank factorization.

\end{abstract}

\thispagestyle{empty}
\clearpage
\setcounter{page}{1}


\section{Introduction}
The goal of  low-rank approximation is to approximate an $n \times n$ matrix $A$ with a rank-$k$ matrix $L$. $L$ can be written as the product $L = U \cdot V$ 
of a ``tall-and-thin'' matrix $U$ and a ``short-and-wide'' matrix $V$ with $k$ columns and rows respectively. For $k \ll n$ this approximation can lead to computational speedups: one can store the factors $U$
and $V$ with less memory than storing $A$ itself, and can compute the product 
$U \cdot V \cdot x$ with a vector $x$ faster than computing $A \cdot x$. Additionally, low-rank approximation is useful for denoising and can reveal low-dimensional structure in high-dimensional data (it is e.g., the basis behind principal component analysis). It thus serves as a preprocessing step in many 
applications, including clustering, data mining, and recommendation systems. The optimal low-rank approximation to $A$ with distance measured in the Frobenius, spectral, or any unitarily invariant norm can be computed in polynomial time using a singular value decomposition (SVD). There are also extremely efficient approximation 
algorithms for finding a near optimal $L$ under different measures, including the Frobenius norm, spectral norm, and various entrywise norms.  
For a comprehensive treatment, 
we refer the reader to the surveys \cite{kv09,m11,w14}. 

Despite its wide applicability, in many situations standard low-rank approximation does not suffice. For example, it is common that certain entries in $A$ either \emph{don't obey underlying low-rank structure} or are  \emph{missing}. For example, $A$ may be close to low-rank but with a small number of corrupted entries, or may be the sum of a low-rank matrix plus a high-rank, but still efficiently representable, diagonal or block diagonal matrix. In both cases, one must compute a low-rank approximation of $A$ ignoring the outlying entries. One can formalize this problem, considering a binary matrix $W$ with $W_{i,j} = 0$ for each outlying entry $(i,j)$ of $A$ and $W_{i,j} = 1$ otherwise.
\begin{problem}[Masked Low-Rank Approximation]\label{def:main}
Given $A \in \R^{n \times n}$, binary $W \in \{0,1\}^{n \times n}$, and rank parameter $k$, find rank-$k$ $L$ minimizing:
$$\|W \circ (A-L)\|_F^2 = \sum_{i,j \in [n]} W_{i,j} \cdot (A_{i,j} - L_{i,j})^2,$$
where for two matrices $M$ and $N$ of the same size, $M \circ N$ denotes the entrywise (Hadamard product): with $(M \circ N)_{i,j} = M_{i,j} \cdot N_{i,j}$ and for integer $n$, $[n]$ denotes $\{1,\ldots, n\}$.
\end{problem}
As stated, Problem \ref{def:main} minimizes the squared Frobenius norm of $W \circ (A-L)$. However any matrix norm can be used. In any case, is unclear how to extend standard low-rank approximation algorithms to solving Problem \ref{def:main}, since they optimize over the full matrix $A$, without the ability to take into account $W$ encoding entries that should be ignored. We note that Problem \ref{def:main} is equivalent to minimizing $\norm{A-(L+S)}_F^2$ where $L$ is rank-$k$ and $S$ is any matrix with support restricted to the $0$ entries of $W$. If these zeros are on the diagonal, then $S$ is diagonal. If they are sparse, then $S$ is sparse, etc. This is how Problem \ref{def:main} is traditionally stated in many applications. 

\subsection{Existing Work}
 A common approach to solving Problem \ref{def:main} is to apply alternating minimization or the EM (Expectation-Maximization)  algorithm. In fact, factor analysis, a slight variant of Problem \ref{def:main} when $W$ is $0$ on its diagonal and $1$ off the diagonal, was one of the original motivations of the EM algorithm \cite{Dempster77,Rubin1982}. 
 Much recent work studies when alternating minimization for Problem \ref{def:main} converges in polynomial time under the assumptions that (1) there is a solution
$L = U \cdot V \approx A$ which is {\it incoherent}, meaning that the squared row norms of $U$
and column norms of $V$ are small and (2) the entries of $W$ are selected at random or have pseudorandom properties \cite{altMinNew,kyrillidis2012matrix,netrapalli}. Under similar assumptions it can be shown that Problem \ref{def:main} and the related problem of \emph{robust PCA} can be solved via convex relaxation  in polynomial time \cite{candes2009exact,wrightNips,candes2011robust}. In many cases, these algorithms perform well in practice even when the above assumptions do not hold. Additionally, they can be proven to run in polynomial time in some common settings when the entries of $W$ are not random -- e.g., when $W$ is zero only on its diagonal or at a few arbitrary locations. That is, when we want to approximate $A$ as a low-rank plus diagonal component, or a low-rank matrix with arbitrary sparse corruptions respectively. However, these results still require assuming the existence of $U \cdot V$ that is incoherent and further that is \emph{exact} -- with $U \cdot V$

A natural question is if for common mask patterns, one can obtain provable algorithms without incoherence or other strong assumptions. This approach
was taken in \cite{razenshteyn2016weighted} in the context of {\it weighted low-rank approximation}, where $W$ is a nonnegative matrix and the objective is still to minimize $\|W \circ (A-L)\|_F^2.$
When $W$ is binary, 
this reduces to Problem \ref{def:main}. 
In \cite{razenshteyn2016weighted} it was shown that if $W$ has at most $r$ distinct columns, then it
is possible to obtain a relative error guarantee in $2^{\poly(rk/\epsilon)} \cdot \poly(n)$ time. More generally, if
the rank of $W$ over the reals is at most $r$, then $n^{\poly(rk/\epsilon)}$ time is achievable. Note that such
algorithms are only polynomial time if $k$, $r$, and $1/\epsilon$ are very small. In many common use cases, such as when $W$ is all $0s$ on the diagonal and $1$ off-diagonal (corresponding to  low-rank plus diagonal decomposition), or when
$W$ is all $0$s above the diagonal and $1$s on or beneath the diagonal, $r$ is large: in fact $\rank(W) = r = n$ in these cases.

When $A$ is low-rank with sparse corruptions, i.e., when $W$ has at most $t$ zero entries per row and column, the algorithms of \cite{razenshteyn2016weighted} can be applied if there is an exact solution (with $A = L$ on all non-corrupted entries). \cite{razenshteyn2016weighted} referred to this problem as \emph{adversarial
matrix completion} and gave an $n^{O(tk^2)}$ time algorithm. This is only polynomial time for constant values of $t$ and $k$, and even for constant $t$ and $k$ is very large. 
Moreover, their method cannot be used in the approximate case 
since it requires creating a low-rank weight matrix $W'$ whose support matches that of $W$. Since $W$ may be far from low-rank, the non-zero entries
of $W$ and $W'$ necessarily have very different values. This introduces significant error, unless $A =L$ exactly on the support of $W$.

\subsection{Our Contributions}

With the goal of obtaining fast masked low-rank approximation algorithms, we consider \emph{bicriteria approximation} with additive
error. That is, we allow  the rank $k'$ of the output $L$ to be slightly larger
than $k$, but one still compares to the best rank-$k$ approximation.  
Formally, given $A \in \R^{n \times n}$, $W \in \{0,1\}^{n \times n}$, and an error parameter $\epsilon$, we
would like to find a rank-$k'$ matrix $L$ for which: 
\begin{eqnarray}\label{eqn:main}
 \|W \circ (A-L)\|_F^2 \leq OPT +  \epsilon \|A\|_F^2,
\end{eqnarray}
where $OPT = \min_{\textrm{rank-}k\ \hat{L}} \|W \circ (A-\hat{L} )\|_F^2$ is the optimal value of Problem \ref{def:main}.

 Assuming a variant of
the Exponential Time Hypothesis, \cite{razenshteyn2016weighted} shows a lower bound of $2^{\Omega(r)}$ time for finding rank-$k$ $L$ achieving \eqref{eqn:main} with constant $\epsilon$ when $W$ is rank-$r$. 
Thus the relaxation to bicriteria approximation seems necessary.
In many applications it is not
essential for the output rank $k'$ to be exactly $k$ -- as long as $k'$ is small, one still obtains 
significant compression. Indeed, bicriteria algorithms for low-rank matrix approximation are widely studied \cite{dv07,FFSS07,cw15focs,chierichetti2017algorithms,song2019relative,bhaskara2019residual}. 
The starting point of our work is the following question:
\begin{center}
{\it For which mask patterns $W \in \{0,1\}^{n \times n}$ can one obtain efficient 
bicriteria low-rank approximation algorithms
with $k' \leq k \cdot \poly((\log n)/\epsilon)$ satisfying \eqref{eqn:main}?}
\end{center}

\paragraph{Main Results:}

We show that the answer to this question is related to the {\it randomized communication complexity} of the mask $W$.\footnote{Our bounds actually hold for the public coin partition number of $W$, which is upper bounded by the randomized communication complexity  \cite{jain2014quadratically}. See Section \ref{sec:appRank} for a more detailed discussion.} If the rows and columns of $W \in \{0,1\}^{n \times n}$ are indexed by strings $x \in \{0,1\}^{\log n}$ and $y \in \{0,1\}^{\log n}$, respectively, we can think of $W$ as a two-player communication matrix for a Boolean function $f$, where $f(x,y) = W_{x,y}$. Here Alice has $x$, Bob has $y$, and the two parties want to exchange messages with as few bits as possible to compute $f(x,y)$ with probability at least $1-\delta$. The number of bits required is the {randomized communication complexity} $R_{\delta}(f)$. If we further require that the protocol never errs when $f(x,y) = 1$, but for any fixed pair $(x,y)$ with $f(x,y) = 0$, the protocol errs with probability at most $\delta$, then the number of bits required is the {$1$-sided randomized communication complexity} $R^{1-sided}_{\delta}(f)$.
We show:
\begin{theorem}\label{thm:1}
Letting $f$ be the function computed by $W \in \{0,1\}^{n \times n}$ and $\neg f$ be its negation, there is a bicriteria low-rank approximation $L$ with rank $k' = k \cdot 2^{R^{1-sided}_{\epsilon}(\neg f)}$ achieving:
\begin{align*}
 \|W \circ (A-L)\|_F^2 \leq OPT +  2\epsilon \|A \circ W\|_F^2,
\end{align*}
where $OPT = \min_{\textrm{rank-}k\ \hat{L}} \|W \circ (A-\hat{L})\|_F^2$. $L$ is computable in $O(\nnz(A)) + n \cdot \poly(k'/\epsilon)$ time. 
\end{theorem}
As we will see, for many common $W$, $R_\epsilon^{1-sided}(\neg f)$ is very small -- with $2^{R_\epsilon^{1-sided}(\neg f)}$ at most $\poly(\log n/\epsilon)$. Note that our additive error is in terms of $\|A \circ W\|_F^2$ which is only smaller than $\|A \|_F^2$, and may be much smaller, if e.g., the zeros in $W$ correspond to corruptions in $A$.
We also show a bound in terms of the communication complexity with 2-sided error.
\begin{theorem}\label{thm:2}
Letting $f$ be the function computed by $W \in \{0,1\}^{n \times n}$, there is a bicriteria low-rank approximation $L$ with rank $k' = k \cdot 2^{R_{\epsilon}(f)}$ achieving:
\begin{align*}
 \|W \circ (A-L)\|_F^2 \leq OPT +  2\epsilon \|A \circ W\|_F^2 + \epsilon \norm{L_{opt} \circ (1-W)}_F^2,
\end{align*}
where $OPT = \min_{\textrm{rank-}k\ \hat{L}} \|W \circ (A-\hat{L} )\|_F^2$ and $L_{opt}$ is any rank-$k$ matrix achieving $OPT$. $L$ is computable in $O(\nnz(A)) + n \cdot \poly(k'/\epsilon)$ time. 
\end{theorem}

Further, the algorithm achieving Theorems \ref{thm:1} and \ref{thm:2} is extremely simple: just zero out the entries in $A$ corresponding to entries in $W$ that are $0$ (i.e., compute  $A \circ W$), and then output a standard rank-$k'$ approximation of the resulting matrix. This is already a widely-used heuristic for solving Problem \ref{def:main} \cite{Azar:2001,zhao2016low}, and we obtain the first provable guarantees. 
An optimal low-rank approximation of $A \circ W$ can be computed in polynomial time via an SVD. An approximation achieving relative error $(1+\epsilon)$ can be computed with high probability in $O(\nnz(A)) + n \cdot \poly(k/\epsilon)$ time, giving the runtime bounds of Theorems \ref{thm:1} and \ref{thm:2} \cite{clarkson2015input}.

\subsubsection{Applications}

Theorems \ref{thm:1} and \ref{thm:2} provide the first bicriteria approximation algorithms for Problem \ref{def:main} with small $k'$ for a number of important special cases of the mask matrix $W$:
\begin{enumerate}
\item If $W$ has at most $t$ zero entries in each row, this is Low-Rank Plus Sparse (LRPS) matrix approximation, which captures the challenge of finding a low-rank approximation when a few entries are not known, or do not obey underlying low-rank structure. It has been studied in the context of adversarial matrix completion \cite{schramm2015low}, robust matrix decomposition \cite{hsu2011robust,candes2011robust}, optics, system identification \cite{beghelli}, and more \cite{chandrasekaran2011rank}.
\item If $W$ is zero exactly on the diagonal entries, this is Low-Rank Plus Diagonal (LRPD) matrix approximation. This problem arises since in practice, many matrices that are not close to low-rank are close to diagonal, or contain a mixture of diagonal and low-rank components \cite{chandrasekaran2011rank}. This observation has been used e.g., to construct compact representations of kernel matrices \cite{snelson2005sparse,Wang:2014}, weight matrices in neural networks \cite{barone2016low,zhao2016low}, and covariance matrices \cite{ten1991numerical,stein}. LRPD approximation also arises in applications related to source separation \cite{liutkus2017diagonal} and variational inference \cite{miller2017variational} and is closely related to  factor analysis \cite{spearman,saunderson2012diagonal}, which adds the additional constraints that $L$ and $A-L$ are PSD. 
\item If $W$ is the negation of a block-diagonal matrix with blocks of varying sizes, meaning that $W$ is $0$ on entries in the blocks and $1$ on entries outside of the blocks, this  is Low-Rank Plus Block-Diagonal (LRPBD) matrix approximation. This is a natural generalization of the LRPD problem and has been studied in the context of anomaly detection in networks \cite{a15}, foreground detection \cite{g12}, and robust principal component analysis \cite{li16}. We also consider the natural generalization of LRPS approximation discussed above, which we call the Low-Rank Plus Block-Sparse (LRPBS) matrix approximation problem.
\item If each row of $W$ has a prefix of an arbitrary number of ones, followed by a suffix of zeros, this is the Monotone Missing Data Pattern (MMDP) problem. This is a common missing data pattern, arising in the event that when a variable is missing from a sample, all subsequent variables are also missing. Methods for handling this pattern are, e.g., included in the SAS/STAT package for statistical analyses \cite{SASSTAT}.
We refer the
reader to \cite{patterns} for more examples of common missing data patterns, such as ``connected'' and ``file matching'' patterns. 
\item If $W$ is the negation of a banded matrix where $W_{i,j} = 0$ iff $|i-j| < p$ for some distance $p$,  this is Low-Rank Plus Banded (LRPBand) matrix approximation. Variants of this problem arise in scientific computing and machine learning, in particular in the approximation of kernel matrices via fast multipole methods \cite{rokhlin1985rapid,greengard1991fast,yang2003improved}. Roughly, these methods approximate a kernel matrix using a low-rank  `far-field' component, 
and a `near-field' component, which explicitly represents the kernel function between close points. If points are in one dimension and sorted, this corresponds to approximating $A$ with a low-rank matrix plus a banded matrix. While many methods compute the low-rank component analytically (using polynomial approximations of the kernel function), a natural alternative is to seek an optimal decomposition via Problem \ref{def:main}.
In many applications, it is common to work in higher dimensions. E.g., in the two-dimensional case, each $i \in [n]$ can be mapped to $(i_1,i_2) \in [\sqrt{n}] \times [\sqrt{n}]$ where $i_1,i_2$ correspond to the first and second halves of $i's$ binary expansion. $W_{i,j} = 0$ iff $|i_1-j_1| + |i_2-j_2| < p$. We give similar bounds for this multidimensional variant.
\end{enumerate}
We summarize our results for the above weight patterns in Table \ref{tab:results}. We give more detail on the specific functions $f$ used in these applications in Sections \ref{sec:pre} and \ref{sec:mainBound}, but note that (1), (2), and (3) use variants of the Equality problem, which has $O(\log(1/\epsilon))$ randomized $1$-sided error communication complexity, (4) and (5) use a variant of the Greater-Than problem with $O(\log \log n + \log(1/\epsilon))$ randomized $2$-sided error communication complexity for $\log n $ bit inputs.

\begin{table}[h]

	\renewcommand{\arraystretch}{1.2}%

	\centering

	\begin{tabular*}{\textwidth}{c @{\extracolsep{\fill}}ccc}

		\hline \hline
		Mask Pattern & $k'$ & Communication Problem & Ref. \\
                \hline
                LRPD/LRPBD & $O(k/\epsilon)$ & Equality & Cors. \ref{cor:lrpd} \& \ref{cor:lrpb} \\
                LRPBand & $k \cdot \poly \left (\frac{\log n}{\epsilon} \right)$ & Variant of Greater-Than & Cors. \ref{cor:band} \& \ref{cor:bandM}\\
                LRPS/LRPBS (w/ sparsity $t$) & $O(k t/\epsilon)$ & Variant of equality & Cors. \ref{cor:lrps} \& \ref{cor:lrpbs}\\
                                                MMDP & $k \cdot \poly \left (\frac{\log n}{\epsilon} \right )$ & Greater-Than & Cor. \ref{cor:mono}\\
                 Subsampled Toeplitz & $O (\min(pk,k/\epsilon))$ & Equality mod $p$ & Cor. \ref{cor:toep}\\
	\end{tabular*}
	\caption{Summary of applications of Theorems \ref{thm:1} and \ref{thm:2}.}\label{tab:results}\vspace{-1em}
\end{table}

\subsubsection{Relation to Matrix Completion}

Masked low-rank approximation is closely related to the well-studied matrix completion problem \cite{candes2009exact,jain2013low,keshavan2010matrix}, however the goal is different. In masked low-rank approximation, we want to approximate $A$ as accurately as possible on the \emph{non-masked entries} (i.e., where $W_{ij} = 1$). In matrix completion, the support of $W$ represents entries in $A$ that are observed and the goal is to approximate $A$ \emph{on the missing entries} (i.e., where $W_{ij} = 0$). The most common approach to solving this problem is in fact to find a low-rank approximation fitting the non-missing entries (i.e., to solve Problem \ref{def:main}), however the two problems are not equivalent. For example, it is not clear that a bicriteria solution to Problem \ref{def:main}, as given by Theorems \ref{thm:1} and \ref{thm:2}, will give anything interesting for the matrix completion problem. In fact, our proof technique implies that it likely will not.

We additionally note that in matrix completion, the mask matrix $W$ is typically assumed to be random and the goal is to recover the missing entries of $A$ when $W$ has as few sampled ones as possible. We do not expect that a random matrix will have low-communication complexity, unless it has further structure (e.g., few zeros or ones per row).

\subsubsection{Other Communication Models}
Theorems \ref{thm:1} and \ref{thm:2} connect communication complexity to the analysis of a simple heuristic for masked low-rank approximation. A natural question is: 
\begin{center}
{\it Can other notions of communication complexity, such as 
multi-party communication complexity, non-deterministic communication complexity, and communication complexity of non-Boolean functions yield algorithms for masked low-rank approximation?}
\end{center}
We answer this question affirmatively.
We first look at multi-party communication complexity, which we show corresponds to masked tensor low-rank approximation.
Here we focus on order-$3$ tensors, though our results are proven for arbitrary order-$t$ tensors.
A tensor is just an array $A \in \mathbb{R}^{n \times n \times n}$. In masked low-rank tensor approximation
we are given such an $A$ and a mask tensor $W \in \{0,1\}^{n \times n \times n}$ and the goal is to find rank-$k$ tensor $L$ minimizing $\norm{W \circ (A-L)}_F^2.$ This problem has been widely studied in the context of low-rank tensor completion \cite{gandy2011tensor,liu2013tensor,mu2014square} and robust tensor PCA \cite{li2015low,lu2016tensor}, which corresponds to the setting where $W$'s zeros represent sparse corruptions of an otherwise low-rank tensor.   Applications include color image and video reconstruction along with low-rank plus diagonal tensor approximation \cite{benner2016reduced}, where $W$ is zero on its diagonal and one everywhere else.
We show:

\begin{theorem}[Multiparty Communication Complexity $\rightarrow$ Tensor Low-Rank Approx]\label{thm:tensorIntro}
Let $f$ be the function computed by $W \in  \{0,1\}^{n \times n \times n}$, $\neg f$ be its negation, and $R^{3,1-sided}_{\epsilon}(\neg f)$ be the randomized $3$-party communication complexity of $\neg f$ in the number-in-hand blackboard model with $1$-sided error. A bicriteria low-rank approximation $L$ with rank $k' = O \left ((k/\epsilon)^2 \cdot 4^{R^{3,1-sided}_{\epsilon}(\neg f)}\right )$ achieving:
\begin{align*}
 \|W \circ (A-L)\|_F^2 \leq OPT +  2\epsilon \|A \circ W\|_F^2,
\end{align*}
where $OPT = \inf_{\textrm{rank-}k\ \hat{L}} \|W \circ (A-\hat{L} )\|_F^2$, can be computed in $O(\nnz(A)) + n \cdot \poly(k/\epsilon)$ time. 
\end{theorem}
We give applications of Theorem \ref{thm:tensorIntro} to the low-rank plus diagonal tensor approximation problem, achieving $k' = O(k^2/\epsilon^4)$ (Cor. \ref{cor:lrpdT}) and the low-rank plus sparse tensor approximation problem achieving $k' = O \left (\frac{k^{2}\cdot t^{4}}{\epsilon^{6}} \right )$ (Cor. \ref{cor:lrpsT}), where $t$ is the maximum number of zeros on any face of $W$.

We also consider 
a common variant of low-rank approximation studied in data mining and information retrieval: {\it Boolean low-rank approximation} (also called binary low-rank approximation). Here one is given binary $A \in \{0,1\}^{n \times n}$ 
 and seeks to find $U \in \{0,1\}^{n \times k}$ and $V \in \{0,1\}^{k \times n}$ minimizing $\|A - U \cdot V\|_0$ where $U \cdot V$ denotes Boolean matrix multiplication and $\norm{\cdot}_0$ is the entrywise $\ell_0$ norm, equal to the squared Frobenius norm in this case. While Boolean low-rank approximation is NP-hard in general \cite{dan2015low,gillis2018complexity}, there is a large body of work studying heuristic algorithms and approximation schemes, when no entries of $A$ are masked \cite{lu2008optimal,shen2009mining,vaidya2012boolean,bringmann2017approximation,fomin2018parameterized}. We show that any black-box algorithm for standard Boolean low-rank approximation yields a bicriteria algorithm for masked Boolean low-rank approximation, with rank depending on the \emph{nondeterministic communication complexity} of the mask matrix $W$.
\begin{theorem}[Nondeterministic Communication Complexity $\rightarrow$ Boolean Low-Rank Approx]\label{thm:nondet}
Let $f$ be the function computed by $W$ and $N(f)$ be the nondeterministic communication complexity of $f$. For any $k' \ge k \cdot 2^{N(f)}$, if one computes $U,V \in \{0,1\}^{n \times k'}$ satisfying $\norm{A \circ W - U \cdot V}_0 \le \min_{\hat U,\hat V \in \{0,1\}^{n \times k'}} \norm{A \circ W - \hat U \cdot \hat V}_0 + \Delta$ then:
$$\norm{W \circ (A- U \cdot V)}_0 \le 2^{N(f)} \cdot OPT + \Delta,$$
where  
$\displaystyle OPT = \min_{\hat U, \hat V \in \{0,1\}^{k \times n}} \norm{W \circ (A- U \cdot V)}_0 $ and  $U \cdot V$ denotes Boolean matrix multiplication.
\end{theorem}
We can apply Theorem \ref{thm:nondet} for example, to the low-rank plus diagonal Boolean matrix approximation problem, where $W$ is zero on its diagonal and one everywhere else. In this case we have $2^{N(f)} = \log n$ and correspondingly $k' = k \log n$ (Cor. \ref{cor:lrpbBool}).
%
%
%


\subsubsection{Connections to Approximate Rank and Other Communication Lower Bounds}\label{sec:appRank}
In Section \ref{sec:techniques} we sketch the proof of Theorem \ref{thm:1}, which is very simple (Theorems \ref{thm:2}, \ref{thm:tensorIntro}, and \ref{thm:nondet} are proved similarly.) The proof is based on covering $W$ with $2^{R^{1-sided}_\epsilon(\neg f)}$ disjoint monochromatic rectangles, which match $W$ on all but a small random subset of its $1$ entries. The existence of a 1-sided error randomized communication protocol for $\neg f$ using ${R^{1-sided}_\epsilon(\neg f)}$ bits of communication is well known to imply the existence of such a covering with $2^{R^{1-sided}_\epsilon(\neg f)}$ rectangles. However, the optimal size of such a covering, which is known as the `public-coin partition bound' \cite{jain2014quadratically}, may be lower than this. In fact, recent work has shown that it is provably smaller for some problems \cite{goos2017randomized}.
Thus, our algorithm can be stated in terms of this bound, giving improved results for these problems. However, as far as we are aware, this bound does not give any improvements for the communication problems we consider (corresponding to natural weight matrices $W$).

The public coin partition bound is a strengthening of the well-studied partition bound \cite{jain2010partition} for randomized communication complexity, which is itself a strengthening of the smooth rectangle bound \cite{jain2010partition}. This logarithm of the smooth rectangle bound is equivalent to the log approximate nonnegative rank of $W$ up to constants \cite{kol2014approximate}. It has been shown that the randomized communication complexity can be polynomially larger than the log partition bound \cite{goos2017randomized}. Additionally, recent work refuting the log approximate rank conjecture \cite{chattopadhyay2019log} has shown that the randomized communication complexity can be exponentially larger than the log approximate nonnegative rank. Thus, improving our results to depend on these communication complexity lower bounds rather than the communication complexity itself would lead to potential improvements for some weight matrices $W$. However, all known separations are for $W$ with complex structure and relatively high communication complexity, and thus not relevant to common applications. Additionally, it is unclear how to extend our techniques to these weaker notions, or to other related notations, such as information complexity \cite{chakrabarti2001informational}. Such extensions would be interesting, e.g., connecting the difficulty of masked low-rank approximation to the approximate rank of the mask.

\subsubsection{Lower Bounds}

Given our results, and the above discussion, a natural question to ask is:
\begin{center}
{\it Is there a natural notion of the complexity of the mask $W$ that characterizes the difficulty of the masked low-rank approximation problem? }
\end{center}
We give some initial results, focused on how communication complexity in particular relates to the best bicriteria approximation factor for masked low-rank approximation achievable in polynomial time. We note that, since our results actually hold with rank depending on the public-coin partition bound \cite{jain2014quadratically}, which has been separated from the randomized communication complexity, the communication complexity itself certainly does not  tightly characterize the difficulty of masked low-rank approximation. However, we view our lower bounds in terms of communication complexity as a step in understanding this difficulty.

We prove two bounds based on a conjecture of the  hardness of approximate $3$-coloring. 
We show that there is a class of masks $W$ such that any polynomial time algorithm achieving guarantee \eqref{eqn:main} and small enough $\epsilon$ requires bicriteria rank $k' = \Omega \left (\frac{D(f)}{\log D(f)}\right )$ where $D(f)$ is the deterministic communication complexity of $f$. Note that $D(f)$ is only greater than $R^{1-sided}_{\epsilon}(\neg f)$ and  $R_{\epsilon}(f)$. 

We strengthen this bound significantly for two natural variants of the masked low-rank approximation problem: when the low-rank approximation $L$ is required to have a non-negative or binary factorization. We note that our techniques yield matching algorithmic results analogous to Theorems \ref{thm:1} and \ref{thm:2} for these variants. 
We show that for these variants on Problem \ref{def:main}, there is a class of masks $W$ such that any polynomial time algorithm achieving guarantee \eqref{eqn:main} for small enough $\epsilon$ requires bicriteria rank which is {exponential in the deterministic communication complexity}, $k' = 2^{\Omega(D(\neg f))}$. This bound matches our algorithmic results for these variants. We note that in the parameter regimes considered (we just require rank $k = 3$), there exist polynomial time algorithms for the \emph{non-masked} versions of binary  and non-negative low-rank approximation. Thus, the hardness in terms of communication complexity comes from adding the mask to the low-rank cost function rather than the binary and non-negativity constraints themselves.

Our lower bounds are closely related to those of \cite{hardt2014computational} on the hardness of bicriteria low-rank matrix completion. We note that for any $n \times n$ mask matrix $W$, we can always bound $D(f) = O(\log n)$. Thus, achieving a $2^{o(D(f))}$ bicriteria approximation factor means achieving an approximation factor sub-polynomial  in $n$. \cite{hardt2014computational} leaves open if achieving a $\sqrt{n}$ bicriteria approximation to rank-$3$ matrix completion is hard (Question 4.3 in \cite{hardt2014computational}), and more generally asks what bicriteria approximation is achievable in polynomial time (Question 4.2 in \cite{hardt2014computational}).

\subsection{Our Techniques}\label{sec:techniques}


 The key ideas behind Theorems \ref{thm:1} and \ref{thm:2} are similar. We focus on Theorem \ref{thm:1} for exposition.
We want to argue that any near optimal rank-$k'$ approximation of $A\circ W$, gives a good bicriteria solution to the masked rank-$k$ approximation problem. For simplicity, here we focus on showing this for the actual optimal rank-$k'$ approximation,
 $L = \argmin_{\rank-k'\ \hat L} \norm{(A \circ W) - \hat L}_F^2$. We show that $\norm{W \circ (A-L)}_F^2  \le OPT + O(\epsilon) \norm{A \circ W}_F^2$ via a comparison method. Namely, 
we  exhibit a rank $k'$ matrix $\bar L$ that:
\begin{enumerate}
\item Nearly matches how well the optimum rank-$k$ solution $L_{opt}$ to Problem \ref{def:main} approximates $A$ on the support of $W$. In particular, $\norm{(A-\bar L)\circ W}_F^2 \le \norm{(A-L_{opt})\circ W}_F^2 + O(\epsilon) \norm{A \circ W}_F^2$.
\item Places \emph{no mass} outside the support of $W$. In particular, $\norm{\bar L \circ (1-W)}_F^2= 0$.
\end{enumerate}
 Since $L$ minimizes the distance to $(A \circ W)$ among all rank-$k'$ matrices, we have $\norm{(A\circ W) - L}_F^2 \le \norm{(A\circ W) - \bar L}_F^2$. However, by (2), $\bar L$ \emph{exactly matches} $A \circ W$ outside the support of $W$ -- both matrices are $0$ there. Thus $L$ must have at least as large error outside the support of $W$, and in turn cannot have larger error on the  support of $W$. That is, we must have $\norm{(A-L) \circ W}_F^2\le \norm{(A-\bar L) \circ W}_F^2$. Then by (1), $L$ satisfies the desired bound.
 
 \subsubsection{From Communication Protocols to Low-Rank Approximations.}
 The key question becomes how to exhibit $\bar L$, which we do using communication complexity. We view $W$ as the communication matrix of some function $f: \{0,1\}^{\log n} \times \{0,1\}^{\log n} \rightarrow \{0,1\}$, with $W_{x,y} = f(x,y)$, where in $f$ we interpret $x,y \in [n]$ as their binary  representations. It is well-known that the existence of a deterministic communication protocol $\Pi$ that computes $f$ with $D(f)$ total bits of communication implies the existence of a partition of $W$ into $2^{D(f)}$ \emph{monochromatic combinatorial rectangles}. That is, there are $2^{D(f)}$ non-overlapping sets $R_i = S \times T$ for $S,T \in [n]$ that partition $W$ and that satisfy $W(R_i)$ is either all $1$ or all $0$. We could construct $\bar L$ by taking the best $k$-rank approximation of each $A(R_i)$ where $R_i$ is colored $1$ (i.e., contains inputs with $f(x,y) = 1$). We could then sum up these approximations to produce $\bar L$ with rank $\le k \cdot 2^{D(f)}$. Note that $\bar L$ is $0$ outside the rectangles colored $1$ -- i.e., outside the support of $W$. Thus condition (2) above is satisfied. Further,
$\bar L$ matches the optimal rank-$k$ approximation on each $R_i$ colored $1$. So it approximates $A$ at least as well as $L_{opt}$ on these rectangles, and since these rectangles fully cover the support of $W$ we have $\norm{(A - \bar L)\circ W}_F^2 \le \norm{(A - L_{opt})\circ W}_F^2$, giving the requirement of (1).
 
 Unfortunately, essentially none of the $W$ that are of interest in applications admit efficient deterministic communication protocols. $k' = k \cdot 2^{D(f)}$ will typically be larger than $n$, giving a vacuous bound. Thus we turn to randomized communication complexity with error probability  $\epsilon$, $R_\epsilon(f)$, which is much lower in these cases. A randomized protocol $\Pi$ achieving this complexity  corresponds to a distribution over partitions of $W$ into $2^{R_\epsilon(f)}$ rectangles. These rectangles are not monochromatic but are close to it -- letting $W_\Pi$ be the communication matrix of the (random) function computed by the protocol, $W_{\Pi}$ is partitioned into $2^{R_\epsilon(f)}$ monochromatic rectangles and further matches $W$ on each $(x,y)$ with probability at least $1-\epsilon$. We prove that, even with this small error, constructing $\bar L$ as above using the partition of $W_\Pi$ instead of $W$ itself gives a solution nearly  matching  $L_{opt}$ up to small additive error. This error will involve $\norm{A \circ W}_F^2$ and $\norm{L_{opt} \circ (1-W)}_F^2$, depending on whether the protocol makes $1$ or $2$-sided error, as seen in Theorems \ref{thm:1} and \ref{thm:2}.
 
\subsubsection{Low-Rank Approximation to $W$ Does Not Suffice.}\label{sec:rebut} A natural view of our argument above is that the existence of an efficient randomized protocol for $W$ implies the existence of a distribution over low-rank matrices (induced by partitions into near monochromatic rectangles) that match $W$ on each entry with good probability. We note that this distributional view is critical -- simply having a low-rank approximation to $W$ matching all but a small fraction of entries does not suffice. The mistaken entries could in the worst case align with very heavy entries of $A$, which must be approximated well to solve masked low-rank approximation to small error. An approximation with small entrywise error (in the $\ell_{\infty}$ sense) would suffice. However, for important cases, e.g., when $W$ is zero on the diagonal and one off the diagonal, such approximations provably require higher rank thank $2^{R_\epsilon(f)}$ and thus relying on them would lead to significantly weaker bounds \cite{alon2009perturbed}.
 
\subsubsection{Other Communication Models}

In extending our results to other communication models, we first consider the connection between multiparty number-in-hand communication and tensor low-rank approximation. Protocols in this model correspond to a partition of the communication tensor $W \in \{0,1\}^{n \times n \times n}$ into $2^{R^3_\epsilon(f)}$ monochromatic (or nearly monochromatic) rectangles of the form $R_i = S \times T \times U$ for $S,T,U \subseteq [n]$, where $R^3_\epsilon(f)$ is the randomized 3-player communication complexity of $W$. We can again argue the existence of a rank $k' = k \cdot 2^{R^3_\epsilon(f)}$ tensor $\bar L$, obtained by taking a near optimal low-rank approximation to each rectangle colored $1$ in $W_\Pi$, which is mostly $0$ outside the support of $W$ and at the same time competes with the best rank-$k$ tensor approximation $L_{opt}$ on the support of $W$. There are different notions of rank for tensors; here we mostly discuss canonical or CP rank. This lets us argue, as in the two player case, that the best rank-$k'$ approximation of $A \circ W$ also competes with $L_{opt}$. It is not known how to find this best rank-$k'$ approximation efficiently, however using an algorithm of \cite{song2019relative} we can find a rank $k'' = O((k'/\epsilon)^2)$ bicriteria approximation achieving relative error $1+\epsilon$. Overall we have $k'' = O\left ( (k/\epsilon)^2 \cdot 2^{2R_\epsilon^3(f)}\right )$, giving Theorem \ref{thm:tensorIntro}.

We next consider the nondeterministic communication complexity. In a nondeterministic communication protocol for a function $f$, players can make ``guesses'' at any point during the protocol $\Pi$. The only requirement is that, (1) for every $x,y$ with $f(x,y) = 1$, for some set of guesses made by the players, the protocol outputs $\Pi(x,y) =1$ and (2) the protocol never outputs $\Pi(x,y) = 1$ for $x,y$ with $f(x,y) = 0$. Such a protocol using $N(f)$ bits of communication corresponds to covering the communication matrix  $W$ with $2^{N(f)}$ \emph{possibly overlapping} monochromatic rectangles. In many cases, the nondeterministic complexity is much lower than the randomized communication complexity. However, for low-rank approximation in the Frobenius norm, the overlap is a problem. We cannot construct $\bar L$ simply  by approximating each rectangle and adding these approximations together. $\bar L$ will be too ``heavy'' where the rectangles overlap. However, for the Boolean low-rank approximation problem, the overlap is less of a problem. We simply  construct $\bar L$ in the same way, letting it be the AND of  the approximations on each rectangle. In the end, we obtain an error bound of roughly $2^{N(f)} \cdot OPT$, owing to the fact that error may still build up on the overlapping sections. Since there are $2^{N(f)} $ rectangles total, each entry is overlapped by at most $2^{N(f)}$ of them. However, since $N(f)$ can be very small, this result gives a tradeoff with Theorems \ref{thm:1} and \ref{thm:2} (which can also be extended to the Boolean case). For example, in Corollary \ref{cor:lrpbBool} we show how to obtain error $\approx O(\log n \cdot OPT)$ for the Boolean low-rank plus diagonal approximation problem, with rank $k' = O(k \log n)$. This is smaller than the $O(k/\epsilon)$ achieved by  Theorem \ref{thm:1} for small $\epsilon$, which may be required to achieve good error if, e.g., $\norm{A}_F^2$ is large.
  
\subsubsection{An Alternative Approach}
In the important cases when $W$ is zero on its diagonal and one elsewhere or has a few non-zeros per row (the low-rank plus diagonal and low-rank plus sparse approximation problems, respectively) the existence of $\bar L$ satisfying the necessary conditions (1) and (2) above can be proven via a very different technique. The key idea is a structural result: that any low-rank matrix cannot concentrate too much weight on more than a few entries of its diagonal, or more generally, on a sparse support outside a few rows. Thus we can obtain $\bar L$ from $L_{opt}$ by explicitly zero-ing out these few large entries falling outside the support of $W$ (e.g., on its diagonal when $W$ has zeros just on its diagonal). We detail this approach in Appendix \ref{app:alt}, giving a bound matching Theorem \ref{thm:1} in this case. We show that the same structural result can also be used to obtain a fixed-parameter-tractable, relative error, non-bicriteria approximation algorithm for Problem \ref{def:main} in the low-rank plus diagonal case, as well as for the closely related factor analysis problem. We are unaware of any  formal connection between this structural result and our communication complexity framework; however, establishing one would be very interesting. 

\subsection{Road Map}
\begin{description}
\item[Section \ref{sec:pre}]  We give preliminaries, defining the communication models we use and giving communication complexity bounds for common mask matrices in these models.
\item[Section \ref{sec:mainBound}] We prove our main results, Theorems \ref{thm:1} and \ref{thm:2}. We instantiate these results for the common mask matrices shown in Table \ref{tab:results}. 
\item[Section \ref{sec:tensor}] We prove Theorem \ref{thm:tensorIntro}, connecting masked tensor low-rank approximation to multiparty  communication complexity. We give examples instantiating the theorem.
\item[Section \ref{sec:nondet}]  We prove Theorem \ref{thm:nondet}, connecting masked Boolean low-rank approximation to nondeterministic communication complexity. We give examples instantiating the theorem.
\item[Section \ref{sec:lower}]  We give two lower bounds exploring the connection between bicriteria rank in masked low-rank approximation and the communication complexity of $W$.
\end{description}


\section{Preliminaries}\label{sec:pre}

\subsection{Notation and Conventions}

Throughout we use $\log z$ to denote the base-2 logarithm of $z$.
For simplicity, so that we can associate any $W \in \R^{n \times n}$ with a function $f: \{0,1\}^{\log n} \times \{0,1\}^{\log n} \rightarrow \{0,1\}$ we assume that $n$ is a power of $2$ and so $\log n$ is an integer. Our results can be easily extended to general $n$. Given a matrix $M \in \R^{n \times n}$ and a combinatorial rectangle $R = S\times T$ for $S,T \subseteq [n]$, we let $M_R$ denote the submatrix of $M$ indexed by $R$. For matrix $M$ we let $1-M$ denote the matrix $N$ with $N_{i,j} = 1-M_{i,j}$. E.g., $1-I$ is the matrix with all zeros on diagonal and all ones off diagonal. 

While in the introduction we focus on low-rank approximation in the Frobenius norm, many of our results will apply to any entrywise matrix norm of the form:
\begin{definition}\label{def:norm} An entrywise matrix norm $\norm{\cdot}_\star: \R^{n \times n} \rightarrow \R$ is a function of the form:
$$\norm{M}_\star = \sum_{i=1}^n \sum_{j=1}^n g(|M_{i,j}|),$$ where $g: \R \rightarrow \R$ is some monotonically increasing nonnegative function.
\end{definition}
\noindent $g(x) = x^2$ gives the squared Frobenius norm, $g(x) = x^p$ gives the entrywise $\ell_p$ norm, $g(x) = 1$  iff $x \neq 0$ gives the entrywise $\ell_0$ norm, etc. See \cite{song2017low,bringmann2017approximation,chierichetti2017algorithms,ban2019ptas} for a discussion of standard low-rank approximation algorithms for these norms.  As discussed, our bicriteria results will simply require applying one of these algorithms to compute a near-optimal low-rank approximation to $A \circ W$ (i.e., $A$ with the masked entries zeroed out).

\subsection{Communication Complexity Models}
We give a brief introduction to the communication models we consider, and refer the reader to the textbooks \cite{kn97,ry19} for more background. We mostly consider two-party communication of Boolean functions, though will discuss extensions to more than two parties below. 

Consider two parties, Alice and Bob, holding inputs $x \in \mathcal{X}$ and $y \in \mathcal{Y}$ respectively. They exchange messages in order to compute a function $f: \mathcal{X} \times \mathcal{Y} \rightarrow \{0,1\}$ evaluated at $(x,y)$. They would like to do this while minimizing the total number of bits exchanged. The communication between the parties is determined by a possibly randomized protocol, which specifies the message of the next player to speak as a function of previous messages received by that player and that player's input. For a given protocol $\Pi$, we let $|\Pi(x,y)|$ denote the number of bits transmitted by the players on inputs $x$ and $y$, and we let $|\Pi| = \max_{x, y} |\Pi(x,y)|$. 

Let $M$ be the communication matrix of $f$, that is, the matrix whose rows are indexed by elements of $\mathcal{X}$ and columns by elements of $\mathcal{Y}$, and for which $M_{x,y} = f(x,y)$. A well known and useful property is that $\Pi$ {\it partitions} $M$ into {\it rectangles} $R = S \times T$, where $S \subseteq \mathcal{X}$ and $T \subseteq \mathcal{Y}$, and every pair $(x,y)$ of inputs with $(x,y) \in S \times T$ has the same output when running protocol $\Pi$. The number of rectangles in the partition is equal to $2^{|\Pi|}$. We call the unique output of $\Pi$ on a rectangle $S \times T$ the {\it label} of the rectangle.

\begin{definition}[Deterministic Communication Complexity]\label{def:det}
The deterministic communication complexity $D(f) = \min_{\Pi} |\Pi|$, where the minimum is taken over all protocols $\Pi$ for which $\Pi(x,y) = f(x,y)$ for every pair $(x,y)$ of inputs. Equivalently, $D(f)$ is the minimum number so that $M$ can be partitioned via a protocol $\Pi$ into $2^{D(f)}$ rectangles for which for every rectangle $R$ and $b \in \{0,1\}$, if $R$ is labeled $b$, then for all $(x,y) \in R$, $f(x,y) = b$. 
\end{definition}
Another notion we need is non-deterministic communication complexity, which can be smaller than the deterministic communication complexity. In a nondeterministic protocol, Alice and Bob are each allowed to make arbitrary `guesses'. If $f(x,y) =1$ the protocol is required to output $1$ for at least some set of guesses. If $f(x,y) = 0$, the protocol should never output $1$, no matter the guesses. Rather than partitioning $M$ into rectangles, a non-deterministic protocol $\Pi$ covers the support of $M$ with a set of at most $2^{|\Pi|}$ possibly overlapping rectangles.
%
\begin{definition}[Non-deterministic Communication Complexity]\label{def:nonDet}
The non-deterministic communication complexity $N(f) = \min_{\Pi} |\Pi|$, where the minimum is taken over all all non-deterministic protocols $\Pi$ computing $f$. Equivalently, $N(f)$ is the minimum number so that $M$ can be covered via a protocol by $2^{N(f)}$ possibly overlapping rectangles such that (1) for every input $(x,y) \in \mathcal{X} \times \mathcal{Y}$ with $f(x,y) = 1$, we have that $(x,y)$ occurs in at least one of these rectangles, and (2) there is no input $(x,y)$ with $f(x,y) = 0$ which occurs in any of these rectangles. 
\end{definition}
We next turn to randomized communication complexity. 
For the purposes of this paper, we will consider public coin randomized communication complexity, i.e., there is a shared random string $r$ that both Alice and Bob have access to. In a randomized protocol $\Pi$, Alice and Bob see $r$ and then run a deterministic protocol $\Pi_r$. We say a protocol $\Pi$ is a $(\delta_1,\delta_2)$-error protocol if for all $x, y \in \mathcal{X} \times \mathcal{Y}$, with $f(x,y) = 1$, $\Pr_r [\Pi_r(x,y) = f(x,y)] \geq 1-\delta_1$ and for all $x, y \in \mathcal{X} \times \mathcal{Y}$ with $f(x,y) = 0$, $\Pr_r [\Pi_r(x,y) = f(x,y)] \geq 1-\delta_2$.  We can then define a general notion of randomized communication complexity:
\begin{definition}[Randomized Communication Complexity -- General]\label{def:randGen}
The $(\delta_1,\delta_2)$-error randomized communication complexity $R_{\delta_1,\delta_2}(f) = \min_{\Pi} |\Pi|$, where the minimum is taken over all $(\delta_1,\delta_2)$-error protocols $\Pi$. Equivalently, $R_{\delta_1,\delta_2}(f)$ is the minimum number so that there is a distribution over protocols inducing partitions of $M$, each containing at most $2^{R_{\delta_1,\delta_2}(f)}$ rectangles, such that (1) for every $(x,y) \in \mathcal{X} \times \mathcal{Y}$ with $f(x,y) = 1$, with probability at least $1-\delta_1$, $(x,y)$ lands in a rectangle which is labeled $1$ and (2) for every $(x,y) \in \mathcal{X} \times \mathcal{Y}$ with $f(x,y) = 0$, with probability at least $1-\delta_2$, $(x,y)$ lands in a rectangle which is labeled $0$.
\end{definition}
Definition \ref{def:randGen} is typically specialized to two cases: the standard randomized communication complexity with $2$-sided error and the randomized communication complexity with $1$-sided error.
\begin{definition}[Randomized Communication Complexity -- $2$-sided]\label{def:rand}
The $\delta$-error randomized communication complexity of $f$ is $R_{\delta}(f) \eqdef R_{\delta,\delta}(f).$
\end{definition}
\begin{definition}[Randomized Communication Complexity -- $1$-Sided]\label{def:1sided}
The $\delta$-error $1$-sided randomized communication complexity of $f$ is $R^{1-sided}_{\delta}(f) \eqdef R_{0,\delta}(f)$.
\end{definition}
In Theorem \ref{thm:1} we consider the 1-sided communication complexity of $\neg f$, which is equal to $R_{\delta,0}(f)$. In some bounds it will be  helpful to also consider the $1$-way  communication complexity:

\begin{definition}[$1$-way Randomized Communication Complexity]\label{def:1way}
The $\delta$-error $1$-way randomized communication complexity $R^{1-way}_{\delta}(f) = \min_{\Pi}|\Pi|$, where the minimum is taken over all $\delta$-error protocols $\Pi$ in which Alice sends a single message to Bob who then outputs the answer. Note that Bob's output bit is included in the communication complexity. Equivalently, $R^{1-way}_{\delta}(f)$ is the minimum number so that there is a distribution on 1-way protocols inducing partitions of $M$ each containing at most $2^{R^{1-way}_{\delta}(f)}$ rectangles, such that (1) for every  $(x,y) \in \mathcal X \times \mathcal Y$, with probability at least $1-\delta$, $(x,y)$ lands in a rectangle  labeled $f(x,y)$ (2) each partition is obtained by considering a partition $P_\mathcal{X}$ of $ \mathcal{X}$ and choosing rectangles of the form $S \times T$ where $S \in P_{\mathcal{X}}$ and $T \subseteq \mathcal{Y}$.
\end{definition}
We note that one can combine these notions in various ways, so one could look at $R^{1-sided, 1-way}_{\delta}$, for example, with the corresponding natural definition. 
We also consider communication with more than $2$ players in the {\it number-in-hand blackboard model}. In this setting there are $t$ players and an underlying $t$-th order communication tensor $M$ with entries corresponding to elements in $\mathcal{X}^1 \times \mathcal{X}^2 \times \cdots \times \mathcal{X}^t$. 
Again, a deterministic number-in-hand protocol partitions $M$ into combinatorial rectangles $A_1 \times A_2 \times \cdots A_t$, where $A_i \subseteq \mathcal{X}^i$ for $i = 1, 2, \ldots, t$, and each rectangle is labeled either $0$ or $1$. See Section \ref{sec:tensor} for a formal definition.

\subsection{Specific Communication Bounds}
We discuss a few problems that will be particularly useful for our applications. We will only need communication upper bounds and in specific models. Note that in this section, as is standard, we state bounds for communication problems with $n$-bit inputs. In our applications to masked low-rank approximation, we will typically apply the bounds when the input size is $\log n$.
\paragraph{Equality}
In the Equality problem, denoted $EQ$, there are two players Alice and Bob, holding strings $x, y \in \{0,1\}^n$, and the function $EQ(x,y) = 1$ if $x = y$, and $EQ(x,y) = 0$ otherwise. 
\begin{theorem}[\cite{kn97}, combining Corollaries 26 and 27 of 
\cite{BCKWY16}]\label{thm:eq}
$R^{1-way}_{\delta}(EQ) \leq (1-\delta) \log((1-\delta)^2/\delta)+5$, and $R^{1-way, 1-sided}_{\delta}(EQ) \leq \log(1/\delta) + 5$. 
\end{theorem}
We also can bound the nondeterminitic communication complexity of inequality, i.e., the function $NEQ(x,y)$ with $NEQ(x,y) = 1$ iff $x \neq y$.
\begin{theorem}\label{thm:neq}
$N(NEQ) \leq \lceil \log n \rceil + 2$.
\end{theorem}
\begin{proof}
Alice simply guesses an index at which $x$ and $y$ differ and sends this index (using $\lceil \log  n \rceil$ bits) along with the value of $x$ at this index to Bob. Bob sends the value of $y$ at this index and the players check if $x$ and $y$ differ at the index.
\end{proof}
Essentially the same protocol can be used to solve the negation of the disjointness problem, with $\neg DISJ(x,y) = 1$ only if there is some $k \in [n]$ with $x(k) = y(k) = 1$. We thus have:
\begin{theorem}[\cite{sherstovLect}]\label{thm:disj}
$N(\neg DISJ) \leq \lceil \log n \rceil + 2$.
\end{theorem}

\paragraph{Greater-Than}
In the Greater-Than problem, denoted $GT$, there are two players Alice and Bob, holding integers $x, y \in \{0, 1, 2, \ldots, n-1\}$, and the function $GT(x,y) = 1$ if $x > y$, and $GT(x,y) = 0$ otherwise.
\begin{theorem}[\cite{n94}]\label{thm:gt}
$R_{\delta}(GT) = O(\log(n/\delta))$. 
\end{theorem}
\paragraph{Equality-Modulo-$p$}
In the Equality-Modulo-$p$ problem, denoted $EQ_p$, there are two players Alice and Bob, holding integers $x, y \in \{0, 1, 2, \ldots, n-1\}$, and the function $EQ_p(x,y) = 0$ if $x-y = 0 \bmod p$, and otherwise $EQ_p(x,y) = 1$. 

\begin{theorem}\label{thm:eqp}
$D(EQ_p) \leq \lceil \log p \rceil + 1$, 
$R^{1-way}_{\delta}(EQ_p) \leq (1-\delta) \log((1-\delta)^2/\delta)+5$, and \\$R^{1-way, 1-sided}_{\delta}(EQ_p) \leq \log(1/\delta) + 5$.
\end{theorem}
\begin{proof}
Note that the players can replace their inputs with $x \bmod p$
and $y \bmod p$ without loss of generality, and the problem is now equivalent
to testing if $x = y$ on $\lceil \log p \rceil$-length bit strings. The
bounds now follow from Theorem \ref{thm:eq}. $D(EQ_p)$ follows since Alice can just send $x \bmod p$ using $\lceil \log p \rceil$ bits and Bob can send the answer using $1$ bit.
\end{proof}

\section{Bicriteria Approximation from Communication Complexity}\label{sec:mainBound}  

In this section we prove our main results, Theorems \ref{thm:1} and \ref{thm:2}, which connect the randomized communication complexity of the binary matrix $W$ to the rank required to solve Problem \ref{def:main} efficiently up to small additive error. We prove a general theorem connecting the rank to $R_{\delta_1,\delta_2}(f)$. Both Theorems \ref{thm:1} and \ref{thm:2} follow as corollaries if we consider the 1-sided error complexity $R_{\delta}^{1-sided}(\neg f) = R_{\delta,0}(f)$ and the 2-sided error complexity $R_{\delta}(f) \eqdef R_{\delta,\delta}(f)$ respectively (Definitions \ref{def:rand} and \ref{def:1sided}).

\begin{theorem}[Randomized Communication Complexity $\rightarrow$ Bicriteria Approximation]\label{thm:rand}
Consider $W \in \{0,1\}^{n \times n}$ and let $f$ be the function computed by it. For $k' \ge k \cdot 2^{R_{\epsilon_1, \epsilon_2}(f)}$, and any entrywise norm $\norm{\cdot }_\star$ (Def. \ref{def:norm}), for any $L$ satisfying $\norm{A \circ W - L}_\star \le \min_{rank-k'\ \hat L} \norm{A \circ W - \hat L}_\star + \Delta$:
\begin{align*}
\norm{(A-L) \circ W}_\star \le OPT + \epsilon_1 \norm{A \circ W}_\star + \epsilon_2 \norm{L_{opt} \circ (1-W)}_\star + \Delta,
\end{align*}
where $OPT = \min_{\rank-k\ \hat L} \norm{(A-\hat L)\circ W}_\star$ and $L_{opt}$ is any rank-$k$ matrix achieving $OPT$.
\end{theorem}
\begin{proof}
As discussed (Def. \ref{def:randGen}), $R_{\epsilon_1,\epsilon_2}(f)$ is the minimum number so that there is a distribution on protocols inducing partitions of $W$, each containing at most $2^{R_{\epsilon_1,\epsilon_2}(f)}$ rectangles, such that (1) for every $x,y \in \{0,1\}^{\log  n}$ with $f(x,y) = 1$, $(x,y)$ lands in a rectangle labeled $1$ with probability $\ge 1- \epsilon_1$ and (2) for every $x,y \in \{0,1\}^{\log  n}$ with $f(x,y) = 0$, $(x,y)$ lands in a rectangle  labeled $0$ with probability $\ge 1-\epsilon_2$. In other words, letting $W_\Pi$ be the (random) matrix corresponding to the function computed by the protocol: (1) $W \circ (1-W_\Pi)$ has each entry equal to $1$ with probability $\le \epsilon_1$ and (2) $W_\Pi \circ (1-W)$ has each entry  equal to $1$ with probability $\le \epsilon_2$.
Thus, fixing some $L_{opt}$:
\begin{align*}
\E_{protocol\ \Pi} \left [\norm{A \circ W \circ (1-W_\Pi)}_\star + \norm{L_{opt} \circ W_\Pi \circ (1-W)}_\star\right ] &\le \epsilon_1 \norm{A \circ W}_\star + \epsilon_2 \norm{L_{opt} \circ (1- W)}_\star.
\end{align*}
Thus, there is at least one protocol $\Pi$ (inducing a partition with $\le 2^{R_{\epsilon_1,\epsilon_2}(f)}$ rectangles) with:
\begin{align}\label{eq:plow}\norm{A \circ W \circ (1-W_\Pi)}_\star + \norm{L_{opt} \circ W_\Pi \circ (1-W)}_\star \le \epsilon_1 \norm{A \circ W}_\star + \epsilon_2 \norm{L_{opt} \circ (1-W)}_\star .
\end{align}
Let $P_1$ be the set  of rectangles on which the protocol achieving \eqref{eq:plow} returns $1$ and $P_0$ be the set on which it returns $0$.
 For any $R \in P_1$ let $L^R = \argmin_{\rank-k\ \hat L} \norm{A_R \circ W_R - \hat L}_\star$ (note that $L^R$ is the size of $R$). Let $\bar L^R$ be the $n \times n$ matrix equal to $L^R$ on $R$ and $0$ elsewhere. Let 
$\bar L = \sum_{R \in P_1} \bar L^R$. Note that $\bar L$ has rank at most $\sum_{R \in P_1} \rank(\bar L^R) \le k \cdot |P_1| \le k \cdot 2^{R_{\epsilon_1, \epsilon_2}(f)}$. Thus, by the assumption that $L$ satisfies $\norm{A \circ W - L}_\star \le \min_{rank-k'\ \hat L} \norm{A \circ W - \hat L}_\star + \Delta$: 
\begin{align}\label{eq:compare}
\norm{(A - L) \circ W}_\star \le \norm{A \circ W-L}_\star &\le \norm{A \circ W-\bar L}_\star + \Delta \nonumber \\
&= \norm{(A \circ W -\bar L) \circ W_\Pi}_\star + \norm{(A \circ W -\bar L) \circ (1-W_\Pi)}_\star+ \Delta \nonumber \\
&=  \norm{(A \circ W -\bar L) \circ W_\Pi}_\star  + \norm{A \circ W \circ (1-W_\Pi)}_\star + \Delta,
\end{align}
where the third line follows since $\bar L$ is $0$ outside the support of $W_\Pi$ (i.e., outside of the rectangles in $P_1$). Since $\bar L$ is equal to the best rank-$k$ approximation to $A_R \circ W_R$ on each rectangle $R$ in $P_1$, and since these rectangles partition the support of $W_\Pi$:
\begin{align*}
\norm{(A \circ W -\bar L) \circ W_\Pi}_\star &\le \norm{(A \circ W -L_{opt}) \circ W_\Pi}_\star \\
&= \norm{(A -L_{opt}) \circ W \circ W_\Pi}_\star + \norm{L_{opt}\circ (1- W) \circ W_\Pi}_\star\\
&\le OPT + \norm{L_{opt}\circ (1- W) \circ W_\Pi}_\star.
\end{align*}
Plugging back into \eqref{eq:compare} and applying \eqref{eq:plow}:
\begin{align*}
\norm{(A - L) \circ W}_\star &\le OPT + \norm{L_{opt}\circ (1- W) \circ W_\Pi}_\star  + \norm{A \circ W \circ (1-W_\Pi)}_\star + \Delta\\
&\le OPT + \epsilon_1 \norm{A \circ W}_\star + \epsilon_2 \norm{L_{opt} \circ (1-W)}_\star + \Delta,
\end{align*}
which completes the theorem.
\end{proof}
\begin{proof}[Proof of Theorems \ref{thm:1} and \ref{thm:2}]
Theorems \ref{thm:1} and \ref{thm:2} follow by applying Theorem \ref{thm:rand} with $\epsilon_1 = \epsilon_2 = \epsilon$ and $\epsilon_1 = \epsilon$, $\epsilon_2 = 0$ respectively, and noting that $ \norm{A \circ W}_\star \le  \norm{A}_\star$ and $\norm{L_{opt} \circ (1-W)}_\star \le \norm{L_{opt}}_\star$.
When $\norm{\cdot }_\star$ is the squared Frobenius norm, $L$ satisfying $\norm{A \circ W - L}_\star \le \min_{rank-k'\ \hat L} \norm{A \circ W - \hat L}_\star + \Delta$ for $\Delta = \epsilon \norm{(A\circ W)-(A \circ W)_{k'}}_F^2 \le \epsilon \norm{A \circ W}_F^2$ can be computed with high probability in $O(\nnz(A)) + n \cdot \poly(k'/\epsilon)$ time.
\end{proof}

\subsection{Applications of Main Theorem}

We now instantiate Theorem \ref{thm:rand} for a number of common mask patterns. See Table \ref{tab:results} for a summary. Note that the additive error bounds achieved are stated in terms of $\norm{A\circ W}_\star$ and $\norm{L_{opt} \circ (1-W)}_\star$, which are only smaller than $\norm{A}_\star$ and $\norm{L_{opt}}_\star$ respectively.
We start with the case when $W$ is the negation of a diagonal matrix or a block diagonal matrix, corresponding to the Low-Rank Plus Diagonal (LRPD) and Low-Rank Plus Block Diagonal (LRPBD) matrix approximation problems.
\begin{corollary}[Low-Rank Plus Diagonal Approximation]\label{cor:lrpd}
Let $W = 1- I$ where $I$ is the $n \times n$ identity matrix. Then for $k' = O \left (\frac{k}{\epsilon} \right )$ and $L$ with $\norm{A \circ W - L}_\star \le \min_{\rank-k'\ \hat L} \norm{A \circ W - \hat L}_\star + \epsilon \norm{A \circ W}_\star$:
\begin{align*}
\norm{(A-L) \circ W}_\star \le OPT + 2 \epsilon \norm{A \circ W}_\star.
\end{align*}
If $\norm{\cdot}_\star = \norm{\cdot }_F^2$, such an $L$ can be computed with high probability in $O(\nnz(A)) + n \poly(k/\epsilon)$ time.
\end{corollary}
\begin{proof}
The function $f$ corresponding to $W$ is the inequality function $NEQ$. We have $R_\epsilon^{1-sided}(\neg NEQ) = R_\epsilon^{1-sided}(EQ) $, which by Theorem \ref{thm:eq} is bounded by $\log(1/\epsilon) + 5$. Thus $2^{R_\epsilon^{1-sided}(\neg NEQ)} \le \frac{32}{\epsilon}$. The corollary then follows directly from Theorem \ref{thm:rand}.
\end{proof}

\begin{corollary}[Low-Rank Plus Block Diagonal Approximation]\label{cor:lrpb} Consider any partition $B_1 \cup B_2 \cup \ldots \cup B_b = [n]$ and let $W$ be the matrix with $W_{i,j} = 0$ if $i,j \in B_k$ for some $k$ and $W_{i,j} = 1$ otherwise.  
Then for $k' = O\left (\frac{k}{\epsilon}\right )$ and $L$ with $\norm{A \circ W - L}_\star \le \min_{\rank-k'\ \hat L} \norm{A \circ W - \hat L}_\star + \epsilon \norm{A \circ W}_\star$:
\begin{align*}
\norm{(A-L) \circ W}_\star \le OPT + 2 \epsilon \norm{A \circ W}_\star.
\end{align*}
If $\norm{\cdot}_\star = \norm{\cdot }_F^2$, such an $L$ can be computed with high probability in $O(\nnz(A)) + n \poly(k/\epsilon)$ time.
\end{corollary}
\begin{proof}
The function $f$ corresponding to $W$ is the inequality function $NEQ$ where $x,y \in [n]$ are identified with $j,k \in [b]$ if block $B_j$ contains $x$ and $B_k$ contains $y$. The randomized communication complexity $\neg f$ is thus bounded by the complexity  of equality. By Theorem \ref{thm:eq}, $R_\epsilon^{1-sided}(EQ) \le \log(1/\epsilon) + 5$ and so $2^{R_\epsilon^{1-sided}(f)} \le \frac{32}{\epsilon}$, which gives the corollary.
\end{proof}

We next consider the Low-Rank Plus Sparse (LRPS) and Low-Rank Plus Block Sparse (LRPBS) approximation problems, where $W$ has at most $t$ nonzeros (or nonzero blocks) per row. Note that this setting strictly generalizes the Low-Rank Plus (Block) Diagonal Problem, and our most general Corollary  \ref{cor:lrpbs} in fact directly implies Corollaries \ref{cor:lrpd}, \ref{cor:lrpb}, and \ref{cor:lrps}.

\begin{corollary}[Low-Rank Plus Sparse Approximation]\label{cor:lrps}
Let $W \in \{0,1\}^{n \times n}$ have at most  $t$ zeros in each row. Then for $k' = O\left (\frac{k t}{\epsilon}\right )$ and $L$ with $\norm{A \circ W - L}_\star \le \min_{\rank-k'\ \hat L} \norm{A \circ W - \hat L}_\star + \epsilon \norm{A \circ W}_\star$:
\begin{align*}
\norm{(A-L) \circ W}_\star \le OPT + 2 \epsilon \norm{A \circ W}_\star.
\end{align*}
If $\norm{\cdot}_\star = \norm{\cdot }_F^2$, such an $L$ can be computed with high probability in $O(\nnz(A)) + n \poly(k t /\epsilon)$ time.
\end{corollary}
\begin{proof}
The function $f$ corresponding to
$W$ is the negation of the problem where Alice is given $x \in [n]$ and must determine if Bob has input $y \in S_x$ where $S_x \subseteq [n]$ has at most $t$ entries, corresponding to the locations of the zero entries in the $x^{th}$ row of $W$. This problem can be solved by running a 1-way equality protocol with error parameter $\epsilon/t$, which by Theorem \ref{thm:eq} requires $R_{\epsilon/t}^{1-sided, 1-sided}(EQ) \le \log(t/\epsilon) + 5$ bits of communication (including Alice's output bit). Alice can then check, with probability $\ge \epsilon/t$ whether Bob's input is equal to each entry in $S_x$. By a union bound, she succeeds in checking if Bob's input is in $S_x$ with probability $\ge 1- \epsilon$. Thus we have $R_{\epsilon}^{1-sided}(\neg f) \le \log(t/\epsilon) + 5$ and so $2^{R_{\epsilon}^{1-sided}(\neg f)} \le \frac{32 t}{\epsilon}$, which completes the corollary.
\end{proof}

As with equality, the block result simply follows from considering the same communication problem as in Corollary \ref{cor:lrps} where $x,y$ are identified with their corresponding blocks. We obtain:
\begin{corollary}[Low-Rank Plus Block Sparse Approximation]\label{cor:lrpbs}
Consider any pair of partitions $B^x_1 \cup \ldots \cup B_b^x = [n]$, and  $B^y_1 \cup \ldots \cup B_b^y= [n]$. Let $W'$ be any $b \times b$ matrix with at most  $t$ zeros in each row. Let $W$ be the $n \times n$ matrix where $W_{i,j} = W'_{k,\ell}$ for $k,\ell$ with $i \in B_k^x$ and $j \in B_\ell^y$. Then for $k' = O \left (\frac{k t}{\epsilon}\right )$ and $L$ with $\norm{A \circ W - L}_\star \le \min_{\rank-k'\ \hat L} \norm{A \circ W - \hat L}_\star + \epsilon \norm{A \circ W}_\star$:
\begin{align*}
\norm{(A-L) \circ W}_\star \le OPT + 2 \epsilon \norm{A \circ W}_\star.
\end{align*}
If $\norm{\cdot}_\star = \norm{\cdot }_F^2$, such an $L$ can be computed with high probability in $O(\nnz(A)) + n \poly(kt/\epsilon)$ time.
\end{corollary}

We now consider common mask patterns  that are not sparse -- $W$  may have a large number of nonzero entries in each row and column. 

\begin{corollary}[Subsampled Toeplitz]\label{cor:toep}
For any integer $p$, let $W \in \{0,1\}^{n \times n}$ be the Toeplitz matrix with $W_{i,j} = 0$ iff $i -j = 0 \mod p$. Then for $k' = O\left (\frac{k}{\epsilon}\right )$ and $L$ with $\norm{A \circ W - L}_\star \le \min_{\rank-k'\ \hat L} \norm{A \circ W - \hat L}_\star + \epsilon \norm{A \circ W}_\star$:
\begin{align*}
\norm{(A-L) \circ W}_\star \le OPT + 2 \epsilon \norm{A \circ W}_\star.
\end{align*}
For $k' \ge 4pk$, $\norm{(A-L) \circ W}_\star \le OPT + \epsilon \norm{A \circ W}_\star$. If $\norm{\cdot}_\star = \norm{\cdot }_F^2$, $L$ satisfying the required guarantee can be computed with high probability in $O(\nnz(A)) + n \poly(k'/\epsilon)$ time.
\end{corollary}
\begin{proof}
The function $f$ corresponding to $W$ is the negation of the equality function mod $p$ with $EQ_p(x,y) = 1$ iff $x - y = 0 \mod p$. By  Theorem \ref{thm:eqp}, $R_\epsilon^{1-sided}(EQ_p) \le \log(1/\epsilon) + 5$ and $R_0^{1-sided}(EQ_p) = D(EQ_p) \le \lceil \log p \rceil + 1$. Thus $2^{R_\epsilon^{1-sided}(EQ_p)} \le \frac{32}{\epsilon}$ and $2^{R_0^{1-sided}(EQ_p)} \le 4p$, giving the corollary.
\end{proof}

Beyond equality, a number of common sparsity  patterns are related to the communication complexity  of the Greater-Than (GT) function, which is bounded by Theorem \ref{thm:gt}. Since two-sided error is required to give efficient GT protocols, we incur an additional error term depending on $L_{opt}$. An interesting question is if this is necessary for efficient bicriteria approximation.
\begin{corollary}[Low-Rank Plus Banded Approximation]\label{cor:band}
For any integer $p \le n$, let $W \in \{0,1\}^{n \times n}$ be the banded Toeplitz matrix with $W_{i,j} = 0$ iff $|i -j| < p$. Then for $k' = k \cdot \min \left (\frac{p}{\epsilon}, \poly \left (\frac{\log n}{\epsilon} \right) \right )$ and $L$ with $\norm{A \circ W - L}_\star \le \min_{\rank-k'\ \hat L} \norm{A \circ W - \hat L}_\star + \epsilon \norm{A \circ W}_\star$:
\begin{align*}
\norm{(A-L) \circ W}_\star \le OPT + 2 \epsilon \norm{A \circ W}_\star + \epsilon \norm{L_{opt} \circ (1-W)}_\star.
\end{align*}
If $\norm{\cdot}_\star = \norm{\cdot }_F^2$, such an $L$ can be computed with high probability in $O(\nnz(A)) + n \poly(k'/\epsilon)$ time.
\end{corollary}
\begin{proof}
The function $f$ corresponding to $W$ is the negation of the AND of $i + p < j$ and $j+p > i$. Thus, it can be solved with two calls to a protocol for Greater-Than (GT). By Theorem \ref{thm:gt}, for $\log n$ bit inputs, $R_\epsilon(GT) = O\left (\log\left (\frac{\log n}{\epsilon}\right )\right )$. Thus $R_\epsilon(f)  =  O\left (\log\left (\frac{\log n}{\epsilon}\right )\right )$ and $2^{R_\epsilon(f)} = \poly \left (\frac{\log n}{\epsilon} \right)$, giving $k' = k \cdot  \poly \left (\frac{\log n}{\epsilon} \right)$. When $p$ is small, we can apply Corollary \ref{cor:lrpbs}, which gives $k' = k \cdot \frac{p}{\epsilon}$, completing the corollary. 
\end{proof}

We also consider a `multi-dimensional' banded pattern. Here each $i \in \{0,1\}^{\log n}$ corresponds to a point $(i_1,i_2)$ in a $\sqrt{n} \times \sqrt{n}$ grid ($i_1$ and $i_2$ are determined by the first $\frac{\log n}{2}$ and last $\frac{\log n}{2}$ bits of $i$ respectively). We focus on the two-dimensional case, although this set up can easily  be generalized to higher dimensions. We can also imagine generalizing to different distance measures over the points $(i_1,i_2)$ using efficient sketching methods (which yield efficient communication protocols) for various distances \cite{bar2004approximating,kane2010exact}.
We have:
\begin{corollary}[Multi-Dimensional Low-Rank Plus Banded Approximation]\label{cor:bandM}
For any $i \in [n]$ let $i_1,i_2 \in [\sqrt{n}]$ be the integers corresponding to the first and last half of its binary expansion.
For any integer $p \le n$, let $W \in \{0,1\}^{n \times n}$ be binary matrix with $W_{i,j} = 0$ iff $\norm{(i_1,i_2)- (j_1,j_2)}_1 < p$. Then for $k' = k \cdot \poly \left (\frac{\log n}{\epsilon} \right)$ and $L$ with $\norm{A \circ W - L}_\star \le \min_{\rank-k'\ \hat L} \norm{A \circ W - \hat L}_\star + \epsilon \norm{A \circ W}_\star$:
\begin{align*}
\norm{(A-L) \circ W}_\star \le OPT + 2 \epsilon \norm{A \circ W}_\star + \epsilon \norm{L_{opt} \circ (1-W)}_\star.
\end{align*}
If $\norm{\cdot}_\star = \norm{\cdot }_F^2$, such an $L$ can be computed with high probability in $O(\nnz(A)) + n \poly(k'/\epsilon)$ time.
\end{corollary}
\begin{proof}
The function $f$ corresponding to $W$ is the predicate $|i_1-j_1| + |i_2 - j_2| \ge p$. We can first run a greater-than protocol to determine with probability $\ge 1- \epsilon/3$ if $i_1 > j_1$ and similarly with probability  $\ge 1-\epsilon/3$ if $i_2 > j_2$. Depending on the outputs of these checks we can evaluate $|i_1-j_1| + |i_2 - j_2| \ge p$ with a third greater-than protocol succeeding with probability at least $1-\epsilon/3$. For example, if both hold, we can check if $i_1 + i_2 < j_1 + j_2 + p$. A union bound gives total success probability at least $1-\epsilon$. By Theorem \ref{thm:gt}, for $\log n$ bit inputs, $R_\epsilon(GT) = O\left (\log\left (\frac{\log n}{\epsilon}\right )\right )$. Thus $R_\epsilon(f)  =  O\left (\log\left (\frac{\log n}{\epsilon}\right )\right )$ and $2^{R_\epsilon(f)} = \poly \left (\frac{\log n}{\epsilon} \right)$, which gives the corollary.
\end{proof}

A similar result holds for low-rank approximation with monotone missing data.

\begin{corollary}[Monotone Missing Data Problem (MMDP)]\label{cor:mono}
Let $W \in \{0,1\}^{n \times n}$ be any matrix where each row of $W$ has a prefix of an arbitrary number of ones, followed by a suffice of zeros. Then for $k' = k \cdot \poly \left (\frac{\log n}{\epsilon} \right)$ and $L$ with $\norm{A \circ W - L}_\star \le \min_{\rank-k'\ \hat L} \norm{A \circ W - \hat L}_\star + \epsilon \norm{A\circ W}_\star$:
\begin{align*}
\norm{(A-L) \circ W}_\star \le OPT + 2 \epsilon \norm{A \circ W}_\star + \epsilon \norm{L_{opt} \circ (1-W)}_\star.
\end{align*}
If $\norm{\cdot}_\star = \norm{\cdot }_F^2$, such an $L$ can be computed with high probability in $O(\nnz(A)) + n \poly(k'/\epsilon)$ time.
\end{corollary}
\begin{proof}
Let $p_x$ be the length of the prefix of ones in the $x^{th}$ row of $W$. Then the function $f$ corresponding to $W$ is $f(x,y) = 1$ iff $p_x \ge y$. That is, it is just the Greater-Than function where Alice maps her input $x$ to $p_x$. Thus by  Theorem \ref{thm:gt}, $R_\epsilon(f) \le R_\epsilon(GT) = O\left (\log\left (\frac{\log n}{\epsilon}\right )\right )$. So $2^{R_\epsilon(f)} = poly \left (\frac{\log n}{\epsilon} \right)$, which gives the  corollary.
\end{proof}

\section{Tensor Low-Rank Approximation from Multiparty Communication Complexity}\label{sec:tensor}
In this section we prove Theorem \ref{thm:tensorIntro}, extending Theorem \ref{thm:rand} to the low-rank approximation of higher order tensors using the number-in-hand multiparty communication model with a shared blackboard. We give applications to natural tensor generalizations of the low-rank plus (block) diagonal and low-rank plus (block) sparse problems. 

We first formally define the communication model we use. 
Consider $t$ players $P_1,\ldots, P_t$ each with access to an input $x_t \in \mathcal{X}_t$. The players would like to compute a function $f: \mathcal{X}_1 \times \mathcal{X}_2 \times \cdots \times \mathcal{X}_t \rightarrow \{0,1\}$. $f$ corresponds to $t^{th}$ order communication tensor $M \in \{0,1\}^{|\mathcal{X}_1| \times \ldots \times |\mathcal{X}_t|}$ with $M_{x_1,\ldots, x_t} = f(x_1,\ldots,x_t)$.
Players exchange messages by writing them on a shared blackboard that all others can see. In a randomized communication protocol $\Pi$, 
players view a string of public random bits $r$. After seeing $r$, the players run a deterministic protocol $\Pi_r$, which specifies the next player to speak as a function of the information written on the blackboard, as well as the message of that player, as a function of what is written on the blackboard and of their input.

We say a protocol $\Pi$ is a $(\delta_1,\delta_2)$-error protocol if for all $(x_1,\ldots, x_t) \in \mathcal{X}_1 \times \cdots \times \mathcal{X}_t$, with $f(x_1,\ldots,x_t) = 1$, $\Pr_r [\Pi_r(x_1,\ldots,x_t) = f(x_1,\ldots,x_t)] \geq 1-\delta_1$ and for all $(x_1,\ldots, x_t) \in\mathcal{X}_1 \times \cdots \times \mathcal{X}_t$ with $f(x_1,\ldots,x_t) = 0$, $\Pr_r [\Pi_r(x_1,\ldots,x_t) = f(x_1,\ldots,x_t)] \geq 1-\delta_2$.  We can then define the multiparty randomized communication complexity:
\begin{definition}[Multiparty Number-in-Hand Randomized Communication Complexity]\label{def:multi}
The $(\delta_1,\delta_2)$-error $t$-party randomized communication complexity $R^t_{\delta_1,\delta_2}(f) = \min_{\Pi} |\Pi|$, where the minimum is taken over all $(\delta_1,\delta_2)$-error protocols $\Pi$. Equivalently, $R^t_{\delta_1,\delta_2}(f)$ is the minimum number so that there is a distribution on protocols inducing partitions of $M$, each containing at most $2^{R^t_{\delta_1,\delta_2}(f)}$ rectangles, such that (1) for every $(x_1,\ldots, x_t) \in \mathcal{X}_1 \times \cdots \times \mathcal{X}_t$ with $f(x_1,\ldots x_t) = 1$, with probability at least $1-\delta_1$, $(x_1,\ldots,x_t)$ lands in a rectangle labeled $1$ and (2) for every $(x_1,\ldots, x_t) \in \mathcal{X}_1 \times \cdots \times \mathcal{X}_t$ with $f(x_1,\ldots,x_t) = 0$, with probability at least $1-\delta_2$, $(x_1,\ldots,x_t)$ lands in a rectangle labeled $0$.
\end{definition}

We now connect the notion of communication complexity  given in Definition \ref{def:multi} to masked tensor low-rank approximation. The rank of any $t^{th}$ order tensor $M$ is the minimum integer $k$ such that $M = \sum_{i=1}^k u_{i,1} \oplus \ldots \oplus u_{i,t}$ for vectors $u_{1,1},\ldots, u_{1,t},\ldots, u_{k,1},\ldots, u_{k,t}$.
All other notions, such as entrywise norm, entrywise product, etc., are generalized in the natural way to $t^{th}$ order tensors. We note that $OPT$ in this setting will be defined as the \emph{infimum} of the error obtained by a low rank tensor approximation $L$. Due to issues of border-rank this infimum may not be achieved by any $L$ (see e.g, \cite{de2008tensor}). However,  we can to efficiently compute a tensor $L$ achieving within small additive error of this infimum, using the algorithms of \cite{song2019relative}. These algorithms apply to the squared Frobenius norm, and other entrywise $\ell_p$ norms, although we focus on the squared Frobenius norm in our runtime bounds.

\begin{theorem}[Multiparty Communication Complexity $\rightarrow$ Tensor Low-Rank Approximation]\label{thm:multiParty}
Consider $t^{th}$ order tensor $W \in \{0,1\}^{n \times\ldots \times n}$ and let $f: \{0,1\}^{\log n} \times \ldots \times \{0,1\}^{\log n} \times \{0,1\}^{\log n} \rightarrow \{0,1\}$ be the function computed by it. 
For $k' \ge k \cdot 2^{R^t_{\epsilon_1, \epsilon_2}(f)}$, and any entrywise norm $\norm{\cdot }_\star$ (Def. \ref{def:norm}) for any $L$ satisfying $\norm{A \circ W - L}_\star \le \inf_{rank-k'\ \hat L} \norm{A \circ W - \hat L}_\star + \epsilon_3 \norm{A \circ W}_\star$:
\begin{align*}
\norm{(A-L) \circ W}_\star \le OPT + (\epsilon_1 + \epsilon_3) \norm{A \circ W}_\star + \epsilon_2 \norm{L_\gamma \circ (1-W)}_\star + \gamma,
\end{align*}
where $OPT = \inf_{\rank-k\ \hat L} \norm{(A-\hat L)\circ W}_\star$ and $L_\gamma$ is any rank-$k$ $t^{th}$ order tensor achieving error $OPT + \gamma $ for $\gamma > 0$. When $\norm{\cdot}_\star$ is the squared Frobenius norm, for any $\epsilon_3 > 0$, $L$ with rank $O((k'/\epsilon)^{t-1})$ satisfying the required bound can be computed with high probability in $O(\nnz(A)) + n \poly(k'/\epsilon)$ time \cite{song2019relative}.
\end{theorem}
Note that the final rank of the approximation $L$ that can be efficiently computed using the algorithms of \cite{song2019relative} is $O((k'/\epsilon)^{t-1}) = O\left ((k/\epsilon)^{t-1} \cdot 2^{(t-1) \cdot R^t_{\epsilon_1, \epsilon_2}(f)}\right )$.
\begin{proof}
The proof closely follows that of Theorem \ref{thm:rand}.
As discussed in Def. \ref{def:multi}, $R^t_{\epsilon_1,\epsilon_2}(f)$ is the minimum number so that there is a distribution on protocols inducing partitions of $W$, each containing at most $2^{R^t_{\epsilon_1,\epsilon_2}(f)}$ rectangles, such that (1) for every $x_1,\ldots, x_t \in \{0,1\}^{\log  n}$ with $f(x_1,\ldots,x_t) = 1$, $(x_1,\ldots,x_t)$ lands in a rectangle labeled $1$ with probability $\ge 1- \epsilon_1$ and (2) for every $x_1,\ldots, x_t \in \{0,1\}^{\log  n}$ with $f(x_1,\ldots,x_t) = 0$, $(x_1,\ldots,x_t)$ lands in a rectangle labeled $0$ with probability $\ge 1-\epsilon_2$. In other words, letting $W_\Pi$ be the (random) binary tensor corresponding to the function computed by the protocol, $W \circ (1-W_\Pi)$ has each entry equal to $1$ with probability at most $\epsilon_1$ and $W_\Pi \circ (1-W)$ has each entry  equal to $1$ with probability  at most $\epsilon_2$.
Thus, fixing some $L_\gamma$:
\begin{align*}
\E_{protocol\ \Pi} \left [\norm{A \circ W \circ (1-W_\Pi)}_\star + \norm{L_\gamma \circ W_\Pi \circ (1-W)}_\star\right ] &\le \epsilon_1 \norm{A \circ W}_\star + \epsilon_2 \norm{L_\gamma \circ (1- W)}_\star.
\end{align*}
Thus, there is at least one protocol $\Pi$ (inducing a partition with $\le 2^{R^t_{\epsilon_1,\epsilon_2}(f)}$ rectangles) with:
\begin{align}\label{eq:plowT}\norm{A \circ W \circ (1-W_\Pi)}_\star + \norm{L_\gamma \circ W_\Pi \circ (1-W)}_\star \le \epsilon_1 \norm{A \circ W}_\star + \epsilon_2 \norm{L_\gamma \circ (1-W)}_\star .
\end{align}
Let $P_1$ be the set  of rectangles (each a subset of $\underbrace{\{0,1\}^{\log n} \times \ldots \times \{0,1\}^{\log n}}_{t-\text{times}}$) on which the protocol achieving \eqref{eq:plowT} returns $1$ and $P_0$ be the set on which it returns $0$.
 For any $R \in P_1$ let $L^R$ be any rank-$k$ tensor satisfying $\norm{A_R \circ W_R - L_R}_\star \le \inf_{\rank-k\ \hat L} \norm{A_R \circ W_R - \hat L}_\star + \frac{\gamma}{2^{R^t_{\epsilon_1, \epsilon_2}(f)}}$. Let $\bar L^R$ be equal to $L^R$ on $R$ and $0$ elsewhere. Let 
$\bar L = \sum_{R \in P_1} \bar L^R$. Note that $\bar L$ has rank at most $k \cdot |P_1| \le k \cdot 2^{R^t_{\epsilon_1, \epsilon_2}(f)}$. So by the assumption that $L$ satisfies $\norm{A \circ W - L}_\star \le \inf_{rank-k'\ \hat L} \norm{A \circ W - \hat L}_\star + \epsilon_3 \norm{A \circ W}_\star$: 
\begin{align}\label{eq:compareT}
\norm{(A - L) \circ W}_\star \le \norm{A \circ W-L}_\star &\le \norm{A \circ W-\bar L}_\star + \epsilon_3 \norm{A\circ W}_\star\nonumber \\
&= \norm{(A \circ W -\bar L) \circ W_\Pi}_\star + \norm{(A \circ W -\bar L) \circ (1-W_\Pi)}_\star+\epsilon_3 \norm{A\circ W}_\star\nonumber \\
&=  \norm{(A \circ W -\bar L) \circ W_\Pi}_\star  + \norm{A \circ W \circ (1-W_\Pi)}_\star + \epsilon_3 \norm{A \circ W}_\star,
\end{align}
where the third line follows since $\bar L$ is $0$ outside the support of $W_\Pi$ (i.e., outside the rectangles in $P_1$). Since $\bar L$ is within additive error $\frac{\gamma}{2^{R^t_{\epsilon_1, \epsilon_2}(f)}}$ of the best rank-$k$ approximation to $A_R \circ W_R$ on each rectangle $R$ in the support of $W_\Pi$ (i.e., each $R \in P_1)$ and since $\norm{(A-L_\gamma)\circ W}_\star = OPT + \gamma$,
\begin{align*}
\norm{(A \circ W -\bar L) \circ W_\Pi}_\star &\le \norm{(A \circ W -L_\gamma) \circ W_\Pi}_\star - \gamma + \frac{\gamma}{2^{R^t_{\epsilon_1, \epsilon_2}(f)}} \cdot |P_1| \\
&\le \norm{(A -L_\gamma) \circ W \circ W_\Pi}_\star + \norm{L_\gamma\circ (1- W) \circ W_\Pi}_\star\\
&\le OPT + \gamma +\norm{L_\gamma\circ (1- W) \circ W_\Pi}_\star.
\end{align*}
Plugging back into \eqref{eq:compareT} and applying \eqref{eq:plowT}:
\begin{align*}
\norm{(A - L) \circ W}_\star &\le OPT +\gamma + \norm{L_\gamma\circ (1- W) \circ W_\Pi}_\star  + \norm{A \circ W \circ (1-W_\Pi)}_\star + \epsilon_3 \norm{A\circ W}_\star\\
&\le OPT + \gamma + \epsilon_1 \norm{A\circ W}_\star + \epsilon_2 \norm{L_\gamma\circ(1-W)}_\star + \epsilon_3 \norm{A\circ W}_\star,
\end{align*}
which completes the theorem.
\end{proof}
Theorem \ref{thm:tensorIntro} follows immediately from Theorem \ref{thm:multiParty} by  considering $R^{t,1-sided}_{\epsilon}(\neg f) = R^t_{0,\epsilon}(f).$

\subsection{Applications of Main Tensor Theorem}

We now give some example applications of Theorem \ref{thm:multiParty}. We focus on a few common settings of $W$, however note that essentially  any of the data patterns considered for matrices in Section \ref{sec:mainBound} can be generalized to the tensor case .
We first consider the natural generalization of the low-rank plus diagonal matrix approximation problem to tensors.
\begin{corollary}[Low-Rank Plus Diagonal Tensor Approximation]\label{cor:lrpdT}
Let $W$ be the $t^{th}$ order tensor with $W_{i_1,\ldots,i_t} = 0$ when $i_1 = i_2 =\ldots = i_t$ and $W_{i_1,\ldots,i_t} = 1$ otherwise. Then for $k' = O \left (\frac{(kt)^{t-1}}{\epsilon^{2t-2}} \cdot 2^{t(t-1)} \right )$ there is an algorithm computing rank-$k'$ $L$ with high probability in $O(\nnz(A)) + n \poly(k'/\epsilon)$ time satisfying:
\begin{align*}
\norm{(A-L) \circ W}_F^2 \le OPT + \epsilon \norm{A \circ W}_F^2.
\end{align*}
\end{corollary}
\begin{proof}
The function $f$ corresponding to $W$ is the inequality function $NEQ$.  By Theorem \ref{thm:eq}, $R_{\epsilon,0}(NEQ)  = R_\epsilon^{1-way,1-sided}(\neg NEQ) = R_\epsilon^{1-way,1-sided}(EQ) $ is bounded by $\log(1/\epsilon) + 5$. In the number-in-hand blackboard model, a single player can simply run this protocol, with error probability $\epsilon' = \epsilon/(t-1)$. The remaining $t-1$ players can then check equality, all succeeding via a union bound with probability $\ge 1-\epsilon$. These players can then all send the results of their equality test and all players can output the solution based on these results. The total communication is $\log(t/\epsilon) + t + O(1)$ and thus we have $2^{R^{t,1-sided}_{\epsilon}(\neg f)} = O\left ( \frac{t}{\epsilon} \cdot 2^t \right).$ Applying Theorem \ref{thm:multiParty}, we can set $k' = O\left ( \frac{kt}{\epsilon} \cdot 2^t \right)$. We can efficiently output $L$ with rank $O\left ((k'/\epsilon)^{t-1}\right ) = O \left (\frac{(kt)^{t-1}}{\epsilon^{2t-2}} \cdot 2^{t(t-1)} \right )$ achieving within an additive $\epsilon \norm{A \circ W}_F^2$ error of the best $k'$-rank approximation to $A \circ W$. Additionally, we can set $\gamma < \epsilon \norm{A \circ W}_F^2$. Since our error is 1-sided we do not pay any error in terms of $\norm{L_\gamma}_F^2$. Overall, we will have $\norm{(A-L)\circ W}_F^2 \le OPT + 3\epsilon \norm{A \circ W}_F^2$, which completes the corollary after adjusting $\epsilon$ by a constant factor.
\end{proof}
As in the case of matrices, Corollary \ref{cor:lrpdT} also immediately extends to the low-rank plus block diagonal tensor approximation problem where $W$ is zero in a block diagonal pattern. We can also use a similar technique to solve the low-rank plus sparse tensor approximation problem. We have:
\begin{corollary}[Low-Rank Plus Sparse Tensor Approximation]\label{cor:lrpsT}
Let $W$ be any $t^{th}$ order binary tensor such that for any fixed $i_1 \in [n]$, the $(t-1)^{th}$ order `face' $W(i_1,\cdot,\ldots,\cdot)$ has at most $s$ zero entries. Then for $k' = O \left (\frac{k^{t-1}s^{(t-1)^2}}{\epsilon^{(t-1)t}} \cdot 2^{6t(t-1)} \right )$ there is an algorithm computing rank-$k'$ $L$ with high probability in $O(\nnz(A)) + n \poly(k'/\epsilon)$ time satisfying:
\begin{align*}
\norm{(A-L) \circ W}_\star \le OPT + \epsilon \norm{A \circ W}_\star.
\end{align*}
\end{corollary}
\begin{proof}
To compute the function $f$ corresponding to $W$, for each of the $s$ entries $(i_2,\ldots,i_t)$ that are zero on the face $W(x_1,\cdot,\ldots,\cdot)$, player 1 must output zero if $x_j = i_j$ for all $j \in 2,\ldots t$. This corresponds to running $t-1$ equality tests for each of the $s$ entries. Each player aside from player $1$ can run the 1-way equality protocol of Theorem \ref{thm:eq} with error $\epsilon'  = \frac{\epsilon}{s}$, requiring $(t-1) \left [\log\left(\frac{s}{\epsilon}\right )+5\right ]$ bits in total. Player 1 can then run the $t-1$ equality tests for the $s$ different zero entries on their face, outputting $0$ if for one of the entries all $t-1$ tests succeed. For any one of the entries, if $(x_2,\ldots,x_t) \neq (i_2,\ldots,i_t)$ the strings differ in at least one location and so this equality test will fail with probability at least $1-\epsilon/s$. If  $(x_2,\ldots,x_t) = (i_2,\ldots,i_t)$, since the equality tests are 1-sided, the protocol will always be correct. Union bounding over all $s$ entries tested, the protocol has 1-sided error at most $\epsilon$. The total communication, including the players' output bits is $R^t_{\epsilon,0}(f) = (t-1) \left [\log\left(\frac{s}{\epsilon}\right )+5\right ] + t$ and thus $2^{R^t_{\epsilon,0}(f)} = O \left (\left (\frac{s}{\epsilon}\right)^{t-1} \cdot 2^{6t} \right ) $.

Applying Theorem \ref{thm:multiParty}, we can set $k' =  \left (k \cdot \left (\frac{s}{\epsilon}\right )^{t-1} \cdot 2^{6t} \right )$
and can efficiently output $L$ with rank $O\left ((k'/\epsilon)^{t-1}\right ) = O \left (\frac{k^{t-1}s^{(t-1)^2}}{\epsilon^{(t-1)t}} \cdot 2^{6t(t-1)} \right )$ achieving within an additive $\epsilon \norm{A \circ W}_F^2$ error of the best $k'$-rank approximation to $A \circ W$. Additionally, we can set $\gamma < \epsilon \norm{A \circ W}_F^2$. Since our error is 1-sided we do not pay any error in terms of $\norm{L_\gamma}_F^2$. Overall, we will have $\norm{(A-L)\circ W}_F^2 \le OPT + 3\epsilon \norm{A \circ W}_F^2$, which completes the corollary after adjusting $\epsilon$ by a constant factor.
\end{proof}

\section{Boolean Low-Rank Approximation from Nondeterministic Communication Complexity}\label{sec:nondet}

In this section we use nondeterministic communication complexity to give a bicriteria approximation result for masked Boolean low-rank approximation.

\begin{reptheorem}{thm:nondet}[Nondeterministic communication complexity $\rightarrow$ Boolean Low-Rank Approximation]
Given $A, W \in \{0,1\}^{n \times n}$, let $f$ be the function computed by $W$. For any $k' \ge 2^{N(f)} \cdot k$, if one computes $U,V \in \{0,1\}^{n \times k'}$ satisfying $\norm{A \circ W - U \cdot V}_0 \le \min_{\hat U,\hat V \in \{0,1\}^{n \times k'}} \norm{A \circ W - \hat U \cdot \hat V}_0 + \Delta$ then:
$$\norm{W \circ (A- U \cdot V)}_0 \le 2^{N(f)} \cdot OPT + \Delta,$$
where  
$\displaystyle OPT = \min_{\hat U, \hat V \in \{0,1\}^{k \times n}} \norm{W \circ (A- U \cdot V)}_0 $ and $U \cdot V$ denotes Boolean matrix multiplication.
\end{reptheorem}
\begin{proof}
As discussed in Definition \ref{def:nonDet}, $N(f)$ is the minimum number so that there is a protocol inducing  a set of $t = 2^{N(f)}$ possibly overlapping rectangles $R_1,\ldots, R_t$ such that for any $x,y \in \{0,1\}^{\log n}$ with $f(x,y) = 1$, $(x,y)$ is in at least one of these rectangles and for any $x,y$ with $f(x,y) = 0$, $(x,y)$ is in none of these rectangles. Let $W_{R_i}$ be the binary matrix that is one on $R_i$ and zero elsewhere. Equivalently, we have $W_{R_1} + \ldots + W_{R_t} = W$ where $+$ denotes Boolean addition. 
Let $\bar U_i, \bar V_i = \argmin_{\hat U, \hat V \in \{0,1\}^{k \times n}} \norm{A \circ W_{R_i} -\hat U \cdot \hat V}_0$. Note that $\hat U \cdot \hat V$ only has support on $R_i$ and is $0$ elsewhere.
Let $\bar U = [U_1,\ldots, U_t]$ and $\bar V = [V_1;\ldots; V_t]$. Note that $\bar U \cdot \bar V$ only has support on $R_1 \cup \ldots \cup R_t$ and is zero wherever $W_{R_1} + \ldots + W_{R_t} = W$ is $0$. Using this fact:
\begin{align*}
\norm{(A-U\cdot V)\circ W}_0 \le \norm{A\circ W -U\cdot V}_0 &\le \norm{A \circ W - \bar U \cdot \bar V}_0 + \Delta\\
&= \norm{(A \circ W - \bar U \cdot \bar V)\circ W}_0 + \Delta.
\end{align*}
We can then bound via triangle inequality (critically using Booleanity here so that we can write $A \circ W = \sum_{i=1}^t A \circ W_{R_i}$):
\begin{align*}
\norm{(A \circ W - \bar U\cdot  \bar V)\circ W}_0 \le \sum_{i=1}^t \norm{A \circ W_{R_i} - \bar U_i \cdot \bar V_i}_0 \le t \cdot OPT,
\end{align*}
which gives the theorem since $t = 2^{N(f)}$.
\end{proof}

\subsection{Applications of Boolean Low-Rank Approximation Theorem}

Using Theorem \ref{thm:nondet} we can give, for example, a bicriteria result for Boolean low-rank plus (block) diagonal approximation. Note that we could also apply Corollary   \ref{cor:lrpb} here, which uses 1-sided randomized communication complexity. The two theorems gives different tradeoffs between rank and accuracy.
\begin{corollary}[Boolean Low-Rank Plus Block Diagonal Approximation]\label{cor:lrpbBool} Consider any partition $B_1 \cup B_2\cup  \ldots \cup B_b = [n]$ and let $W \in \{0,1\}^{n \times n}$ be the block diagonal  matrix with $W_{i,j} = 0$ if $i,j \in B_k$ for some $k$ and $W_{i,j} = 1$ otherwise.  
Then for $k' \ge 8 k \lceil \log b\rceil$ and $U,V$ satisfying $\norm{A \circ W - U\cdot V}_0 \le \min_{\hat U , \hat V \in \{0,1\}^{k' \times n}} \norm{A \circ W - \hat U\cdot  \hat V}_0 + \Delta$:
\begin{align*}
\norm{(A-U\cdot V) \circ W}_0 \le 8 \lceil \log b\rceil \cdot OPT + \Delta,
\end{align*}
where $\displaystyle OPT = \min_{\hat U, \hat V \in \{0,1\}^{k \times n}} \norm{ (A- U \cdot V)\circ W}_0 $  and $U \cdot V$ denotes Boolean matrix multiplication.
\end{corollary}
Note that if $W$ is simply $1-I$, corresponding to the standard low-rank plus diagonal approximation problem, $k' = O(k \log n)$ and the approximation factor is $O(\log n)$.
\begin{proof}
$W$ is the communication matrix of the inequality  problem NEQ where Alice and Bob first map their inputs to the index of the block containing them. By Theorem \ref{thm:neq} we have $N(f) \le \lceil \log\left ( \lceil \log b\rceil\right ) \rceil + 2$ and so $2^{N(f)} \le 8 \lceil \log b \rceil$, which gives the corollary.
\end{proof}
We note that, as in Corollary \ref{cor:toep}, a similar bound can be given where $W$ is the binary Toeplitz matrix corresponding to inequality mod $p$. 

In some cases, the nondeterministic communication complexity can be much lower than the randomized communication complexity, allowing us to obtain much tighter bicriteria bounds. For example we can consider a sparsity pattern corresponding to the disjointness function:
\begin{corollary}\label{cor:kdisj} Let $W \in \{0,1\}^{n \times n}$ have $W_{i,j} = 0$ if, letting $x,y \in \{0,1\}^{\log n}$ be the binary representations of $i,j$ respectively, there is no $k$ for which $x(k) = y(k) = 1$ (i.e., $x$ and $y$ are disjoint). Otherwise, let $W_{i,j} = 1$.
Then for $k' \ge 8 k \log n$ and $U,V$ satisfying $\norm{A \circ W - U\cdot V}_0 \le \min_{\hat U , \hat V \in \{0,1\}^{k' \times n}} \norm{A \circ W - \hat U\cdot  \hat V}_0 + \Delta$:
\begin{align*}
\norm{(A-U\cdot V) \circ W}_0 \le 8 \log n \cdot OPT + \Delta
\end{align*}
where $\displaystyle OPT = \min_{\hat U, \hat V \in \{0,1\}^{k \times n}} \norm{ (A- U \cdot V)\circ W}_0 $ and $U \cdot V$ denotes Boolean matrix multiplication.
\end{corollary}
\begin{proof}
$W$ is the communication matrix of the negation of the disjointness function $\neg DISJ$ on $\log n$ bit strings, which by Theorem \ref{thm:disj} has communication complexity $N(\neg DISJ) = \lceil \log(\log n) \rceil + 2$. We thus have $2^{N(f)} \le 8 \log n$, which yields the corollary.
\end{proof}
The randomized communication complexity of $W$ in Corollary \ref{cor:kdisj} is the same as the randomized complexity for set disjointness, which is $\Theta(\log n)$ \cite{kalyanasundaram1992probabilistic} on $\log n$ bit inputs. Plugging this complexity e.g. into Theorem \ref{thm:rand} would thus require rank $k' = \poly(n)$.

\section{Lower Bounds}\label{sec:lower}

As discussed, solving Problem \ref{def:main}, even up to additive error $\Theta(1) \cdot \norm{A}_F^2$ (i.e. achieving \eqref{eqn:main} with $\epsilon = \Theta(1)$), was shown by  \cite{razenshteyn2016weighted} to require $2^{\Omega(r)}$ time when $W$ is rank-$r$. Essentially all weight matrices of interest are not low-rank and so, to solve Problem \ref{def:main} in polynomial time, it seems that resorting to bicriteria approximation is necessary.
A natural open question is: what bicriteria rank $k'$ is required? Can this rank be characterized by some natural measure of $W$'s complexity? Our main Theorems \ref{thm:1} and \ref{thm:2} can be shown to hold with rank $k'$ equal to the $k$ times the public coin partition bound of W \cite{jain2014quadratically}. The log of this bound can be polynomially smaller than the randomized communication complexity \cite{goos2017randomized}. Thus, the communication complexity itself does not tightly characterize the bicriteria rank. However, for some classes of weight matrices, we can lower bound the bicriteria rank in terms of the communication complexity. We view such lower bounds as a first step in better understanding the hardness of bicriteria masked low-rank approximation.
We give two results, based on the following conjecture on the hardness of approximate $3$-coloring:
\begin{conjecture}[Hardness of Approximate $3$-coloring]\label{conj:main}
For some fixed $\gamma > 0$, there is no polynomial time algorithm that given a $3$-colorable graph $G$ on $n$ nodes returns a valid $n^\gamma$ coloring of $G$.
\end{conjecture}
It is known that chromatic number of a graph in general is hard to approximate beyond an $n^\gamma$ factor \cite{zuckerman2006linear}. While a hardness result beyond $\Omega(1)$ has not been shown when the chromatic number is  $3$ \cite{dinur2009conditional}, a long line of work on approximate $3$-coloring has failed to break the $n^\gamma$ approximation barrier \cite{wigderson1983improving,blum1994new,karger1998approximate}, with the smallest known $\gamma$ currently $0.207$ \cite{arora2006new}. This line of work uses a relaxation-based approach to the problem and there is some evidence that this approach cannot go beyond a polynomial approximation factor \cite{szegedy1994note,feige2004graphs,dinur2009conditional}. A related but weaker conjecture on the hardness of coloring was used to show hardness of bicriteria approximate matrix completion in \cite{hardt2014computational}. Finally, we note that even if Conjecture \ref{conj:main} does not hold, our lower bounds still show that improving the bicriteria approximation factor for masked low-rank approximation would lead to a breakthrough in the $3$-coloring problem: the discovery of a polynomial time algorithm with a sub-polynomial  approximation factor.
Assuming Conjecture \ref{conj:main}, we show that there is a class of weight matrices such that:

\begin{enumerate}
\item A near-linear dependence of the bicriteria rank on the deterministic communication complexity, $k' = \Omega \left (\frac{D(f)}{\log D(f)}\right )$ is required to obtain polynomial runtime, even to achieve within additive error $\epsilon \cdot \norm{A}_F^2$ of OPT for small enough $\epsilon$. Note that $D(f) \ge R_{\epsilon}^{1-sided}(\neg f)$ and so this bound is only stronger than a near-linear bound in terms of $R_{\epsilon}^{1-sided}(\neg f)$.
\item An exponential dependence of the bicriteria rank on the deterministic communication complexity, $k' = 2^{\Omega(D(f))}$ is required for two natural variants of the masked low-rank approximation problem -- when the low-rank approximation $L$ is required to have a non-negative or binary factorization. In the binary case, the bound holds even for algorithms that achieve within additive error $\epsilon \cdot \norm{A}_F^2$ of OPT for any constant $\epsilon < 1$. 

We note that our techniques yield matching algorithmic results analogous to Theorems \ref{thm:1} and \ref{thm:2} for these variants. We also note that in the parameter regimes considered (we just require $k = 3$), there exist polynomial time algorithms for the \emph{non-masked} versions of these variants. Thus, the hardness in terms of communication complexity comes from adding the mask to the low-rank cost function rather than the binary and non-negativity constraints themselves.
\end{enumerate}

Our lower bounds are closely related to those of \cite{hardt2014computational} on the hardness of bicriteria low-rank matrix completion. We note that for any $n \times n$ mask matrix $W$, we can always bound $D(f) = O(\log n)$. Thus, achieving a $2^{o(D(f))}$ bicriteria approximation factor means achieving an approximation factor sub-polynomial  in $n$. \cite{hardt2014computational} leaves open if achieving a $\sqrt{n}$ bicriteria approximation to rank-$3$ matrix completion is hard (Question 4.3 in \cite{hardt2014computational}), and more generally asks what bicriteria approximation is achievable in polynomial time (Question 4.2 in \cite{hardt2014computational}). 

\subsection{Exponential Communication Complexity Lower Bound for Binary and Non-Negative Masked Low-Rank Approximation}\label{sec:expLB}

We start with our lower bound for two variants of Problem \ref{def:main} where the low-rank approximation $L$ is required to have a binary or non-negative factorization. Binary and non-negative matrix factorization are both well-studied problems in the non-masked setting \cite{lee2001algorithms,ding2005equivalence,arora2012computing,shen2009mining,dan2015low, fomin2018parameterized,kumar2019faster} and our algorithmic results here may be of independent interest.
\begin{problem}[Masked Binary Low-Rank Approximation]\label{def:main2}
Given $A \in \R^{n \times n}$, binary $W \in \{0,1\}^{n \times n}$, and rank parameter $k$, find binary $U,V \in \{0,1\}^{n \times k}$  minimizing:
$$\|W \circ (A-UV^T)\|_F^2 = \sum_{i,j \in [n]} W_{i,j} \cdot (A_{i,j} - (UV^T)_{i,j})^2.$$
\end{problem}
\begin{problem}[Masked  Non-Negative Low-Rank Approximation]\label{def:main3}
Given $A \in \R^{n \times n}$, binary $W \in \{0,1\}^{n \times n}$, and rank parameter $k$, find $U,V \in \R^{n \times k}$  with non-negative entries minimizing:
$$\|W \circ (A-UV^T)\|_F^2 = \sum_{i,j \in [n]} W_{i,j} \cdot (A_{i,j} - (UV^T)_{i,j})^2.$$
\end{problem}

Both these variants admit an efficient bicriteria solution for constant $k$ using similar arguments to our previous bounds (in particular, Theorem \ref{thm:rand}). 

\begin{theorem}[Masked Binary Low-Rank Approximation]\label{thm:binary1}
Consider $W \in \{0,1\}^{n \times n}$ and let $f$ be the function computed by it. Assuming knowledge of a randomized communication protocol for $\neg f$ achieving complexity $R_{\epsilon}^{1-sided}(\neg f)$, there is an algorithm running in $2^{R_{\epsilon}^{1-sided}(\neg f) \cdot k^2 \log k} \cdot \poly(n)$ time that outputs, $U,V \in \{0,1\}^{n \times k'}$ with $k' = k \cdot 2^{R_{\epsilon}^{1-sided}(\neg f)}$ satisfying with high probability:
\begin{align*}
\norm{(A-UV^T) \circ W}_F \le O(1) \cdot OPT + \epsilon \norm{A}_F,
\end{align*}
where $OPT$ is the optimum value of Problem \ref{def:main2}. Since $2^{R_{\epsilon}^{1-sided}(\neg f)} \le n$ for all $W$, the runtime is polynomial when $k = O(1)$.
\end{theorem}
\begin{proof}
Since we have knowledge of a communication protocol for $\neg f$, we can explicitly compute an $O(1)$-optimal binary rank-$k$ approximation of each monochromatic rectangle colored $1$ by the protocol. Each computation requires $2^{O(k^2 \log k)} \poly(n)$ time using the algorithm of \cite{kumar2019faster}. By the argument used in Theorem \ref{thm:rand} the given error bound will hold in expectation over the randomized protocol, and can be achieved with high probability by repeating the algorithm a logarithmic number of times and choosing the best approximation found.
\end{proof}

\begin{theorem}[Masked Non-Negative Low-Rank Approximation]\label{thm:nonneg}
Consider $W \in \{0,1\}^{n \times n}$ and let $f$ be the function computed by it. Assuming knowledge of a randomized communication protocol for $\neg f$ achieving complexity $R_{\epsilon}^{1-sided}(\neg f)$, there is an algorithm running in $2^{R_{\epsilon}^{1-sided}(\neg f)} \cdot n^{O(k^2)}$ time  time that outputs non-negative $U,V \in \R^{n \times k'}$ with $k' = k \cdot 2^{R_{\epsilon}^{1-sided}(\neg f)}$ satisfying with high probability:
\begin{align*}
\norm{(A-UV^T) \circ W}_F \le \epsilon \norm{A}_F,
\end{align*}
when the optimum value of Problem \ref{def:main3} is $OPT = 0$. Since $2^{R_{\epsilon}^{1-sided}(\neg f)} \le n$ for all $W$, the runtime is polynomial when $k = O(1)$.
\end{theorem}
\begin{proof}
The proof is analogous to that of Theorem \ref{thm:binary1}. Each rank-$k$ non-negative low-rank approximation for a monochromatic rectangle can be computed in $n^{O(k^2)}$ via \cite{moitra2012singly}. This algorithm requires an exact non-negative factorization to exist, which it does for each monochromatic rectangle when $OPT = 0$ since we have a 1-sided error protocol for $\neg f$. We could apply the results of \cite{arora2012computing} to the case when $OPT > 0$ and/or two sided error is allowed in the communication protocol, although the bounds are somewhat more complex.
\end{proof}

We now show that the bicriteria factor of Theorems \ref{thm:binary1} and \ref{thm:nonneg} cannot be improved significantly under Conjecture \ref{conj:main}. For Theorem \ref{thm:binary1} we give a general lower bound, applying for any additive error $\epsilon \norm{A}_F^2$ with constant $\epsilon < 1$. For Theorem \ref{thm:nonneg}, our lower bound only applies to the case where $\epsilon = 0$. Our lower bounds are in terms of the deterministic communication complexity of $f$, which is only higher than $R_{\epsilon}^{1-sided}(\neg f)$.

\begin{theorem}[Masked Binary/Non-Negative Low-Rank Approximation Lower Bound] \label{thm:lb1} Assuming Conjecture \ref{conj:main}, there is no polynomial time algorithm achieving the guarantee of Theorem \ref{thm:binary1} with and constant $\epsilon < 1$ and rank $k' = 2^{o(D(f))}$. There is also no polynomial time algorithm achieving the guarantee of Theorem \ref{thm:nonneg} with $\epsilon = 0$ and rank $k' = 2^{o(D(f))}$.
\end{theorem}
\begin{proof}
We follow the reduction in \cite{hardt2014computational} from coloring to matrix completion (i.e., masked low-rank approximation).
Consider any undirected  $n$-node graph $G = (V,E)$, let $A$ be the $n \times n$ identity matrix, and let $W$ be the $n \times n$ mask matrix with $W_{ii} = 1$ for all $i$, $W_{i,j} = 1$ for all $(i,j) \in E$, and $W_{ij} = 0$ for $(i,j) \notin E$.
We can see that any exact masked low-rank approximation $L$ for $A$ with mask $W$ must admit a factorization $UV^T$ where $u_i^T v_i = 1$ for all $i$ and $u_i^T v_j = 0$ for all $j$ that are neighbors of $i$ in $G$. Since $G$ is $3$-colorable, one valid factorization is the binary (and hence also non-negative) factorization where $U = V$ and each row $u_i$ of $U$ is the standard basis vector corresponding to the color assigned to node $i$. Thus for both Problems \ref{def:main2} and \ref{def:main3}, $OPT = 0$.

We first show that finding a rank $k' = 2^{o(D(f))}$ binary or non-negative factorization achieving error $OPT = 0$ is hard under Conjecture \ref{conj:main}. We then make the result robust to additive $\epsilon \norm{A}_F^2$ error in the binary case. 
Consider binary $U,V \in \R^{n \times k'}$ achieving cost $OPT = 0$. Note that $u_i$ and $v_i$ must overlap on exactly $1$ entry to have $u_i^T v_i = 1$ as required (since $A_{ii} = W_{ii} = 1$ for all $i$). Let $Z \in \R^{n \times k'}$ be the matrix whose $i^{th}$ row $z_i = u_i \circ v_i$ just contains this one overlapping entry. $Z$ is clearly efficiently computable from $U,V$, has $z_i^T z_i = 1$ for all $i$, and if $u_i^T v_j = 0$, clearly, has $z_i^T z_j = 0$ since $Z$'s rows only have fewer non-zeros. Thus, $ZZ^T$ is a binary factorization of $A$ also achieving zero error under the mask $W$. Further, each row of $Z$ has just a single non-zero, which corresponds to a color for the $i^{th}$ node in $G$. Thus, $Z$ directly gives a valid $k'$ coloring of $G$. If $k' = 2^{o(D( f))} = 2^{o(\log n)}$ this would refute Conjecture \ref{conj:main}.

A nearly  identical argument holds for masked non-negative matrix factorization. Consider non-negative $U,V \in \R^{n \times k'}$ achieving cost $OPT = 0$ and let $\bar U, \bar V$ have each entry equal to $0$ where $U,V$ are $0$ and equal to $1$ where $U,V$ are $> 0$. $\bar U, \bar V$ would give a binary  factorization with zero error, except that we may have $\bar u_i^T \bar v_i > 1$ for some $i$. However, if we let $Z \in \R^{n \times k'}$ be the matrix whose $i^{th}$ row $z_i$ is given by picking a single entry to $\bar u_i \circ \bar v_i$ and setting the rest to $0$, we will have $z_i^T z_i = 1$ for all $i$ and $z_i^T z_j = 0$ for all $(i,j) \in E$. Thus $Z$ is a binary factorization of $A$ also achieving zero error under the mask $W$ and as before gives a valid $k'$ coloring of $G$. If $k' = 2^{o(D(f))} = 2^{o(\log n)}$ this would refute Conjecture \ref{conj:main}.

Finally, we make the above argument robust to additive $\epsilon \norm{A}_F^2$ error  in the binary case. Since $A$ is just the identity matrix,  and since $U,V$ are binary, additive $\epsilon \norm{A}_F^2$ error implies that $U,V$ do not match $A$ in at most $\epsilon n$ locations. Let $S$ be the set of at most $\epsilon n$ rows on which $U,V$ err at least once. Let $G'$ be the graph with the nodes corresponding to those rows removed. Then as described above, $U,V$ can be used to compute a valid $k'$ coloring of $G'$. Assuming that $\epsilon < 1$, repeating the process $\log_{1-\epsilon} 1/n $ times on the remaining uncolored nodes, since the size of the input is cut by $1-\epsilon$ each time, gives a $\log_{1-\epsilon}(1/n) \cdot k'$ coloring of the full graph $G$. For constant $\epsilon$, if $k' = 2^{o(D(f))} = 2^{o(\log n)}$ this would refute Conjecture \ref{conj:main}.
\end{proof}

\subsection{Near Linear Communication Complexity Lower Bound for Masked Low-Rank Approximation}
We next show a near-linear lower bound on the bicriteria rank in terms of communication complexity for the masked low-rank approximation problem (Problem \ref{def:main}). We conjecture that this bound can be improved to exponential in communication complexity, like the bounds given in Section \ref{sec:expLB}.

Our bound uses a result of \cite{hardt2014computational} which shows that a real valued low-rank factorization for the coloring matrix described in the proof of Theorem \ref{thm:lb1} can be rounded to obtain a large independent set of $G$, which via repetition will yield a small coloring. There is a loss in the size of the coloring due to this rounding scheme, which is why the eventual lower bound is weaker than that of Section \ref{sec:expLB}, only near-linear rather than exponential in the communication complexity.
\begin{lemma}[Lemma 2.2 of \cite{hardt2014computational}]\label{lem:22}
Consider an $n$-node graph $G = (V,E)$. Let $A$ be the $n \times n$ identity matrix, and  $W$ be the $n \times n$ mask matrix with $W_{ii} = 1$ for all $i$, $W_{i,j} = 1$ for all $(i,j) \in E$, and $W_{ij} = 0$ for $(i,j) \notin E$. Given rank-$k'$ $L$ with $\norm{(A-L) \circ W}_F^2 \le \epsilon \norm{A}_F^2 = \epsilon n$ and $|L_{ij}| < c$ for all $i,j$, there is a randomized polynomial time algorithm that finds an independent set in $G$ of size 
\begin{align*}
T \ge \frac{(1-4(ck')^2 \epsilon) n}{k' \sqrt{\pi} (8 \sqrt{ck'})^{k'}}.
\end{align*} 
\end{lemma}

\begin{theorem}[Masked Low-Rank Approximation Lower Bound] \label{thm:lb3} Assuming Conjecture \ref{conj:main}, there is no polynomial time algorithm achieving the guarantee of \eqref{eqn:main} for Problem \ref{def:main} with $\epsilon \le \frac{1}{8(ck')^2}$, that outputs $L$ with $|L_{ij}| < c$ for all $i,j$ and rank $k' = o \left (\frac{D(f)}{\log D(f)} \right )$. 
\end{theorem}
We note that the lower bound applies to $\epsilon$ that depends on the size of the maximum entry in $L$. Our upper bounds (Theorem \ref{thm:rand} in particular) do not restrict the maximum entry of $L$. However, our hard instance  will be binary (using $A$ and $W$ from Lemma \ref{lem:22}) with $OPT = 0$ and thus our algorithms can always output $L$ with $c = \poly(k)$ bounded entries (using e.g. a result like Lemma 2.1 of \cite{hardt2014computational}). 
\begin{proof}
Letting $A$ and $W$ be as defined in Lemma \ref{lem:22}, the conditions on $L$ ensure that, by Lemma \ref{lem:22}, given $L$ a polynomial time algorithm can find an independent set in $G$ of size
\begin{align*}
T \ge \frac{n/2}{k' \sqrt{\pi} (8 \sqrt{ck'})^{k'}}.
\end{align*}
If $k' = o \left (\frac{D(f)}{\log D(f)} \right ) = o \left (\frac{\log n}{\log \log n} \right ) $ then the denominator of this fraction is $o(n^\gamma)$ for any constant $\gamma$. Thus, the independent set identified has size $\omega(n^{1-\gamma})$. If we repeat the algorithm $o(n^{\gamma})$ times, each time removing the independent set found and assigning a color to all nodes in this set, we will color the full graph with  $o(n^{\gamma})$ colors. Since the runtime is clearly polynomial, this contradicts Conjecture \ref{conj:main}.
\end{proof}

\section{Open Questions}

By focusing on bicriteria approximation, we show how to solve masked low-rank approximation in polynomial time using a  simple heuristic. A number of open questions remain. It would be very interesting to improve the bicriteria ranks we achieve for common masks (summarized in Table \ref{tab:results}). It would also be interesting to give relative error bounds achieving error $(1+\epsilon) \cdot OPT$ instead of our additive error bounds. This is challenging sinces it requires achieving zero error when there is an exact masked low-rank factorization of $A$. 

Relatedly, while we have connected bicriteria masked low-rank approximation to the randomized communication complexity of the mask matrix $W$ (in fact, the public coin partition number of $W$), it would be very interesting to find a notion of $W$'s complexity that tightly characterizes the bicriteria rank achievable in polynomial time. We make some initial steps in Section \ref{sec:lower}, but the question remains mostly unanswered.

Finally, a related problem is \emph{weighted low-rank approximation} -- when $W$ is real valued and we seek to minimize $\norm{W \circ (A-L)}_F^2$. Approximation algorithms depending exponentially on the rank $k$, error parameter $\epsilon$, and various notions of $W$'s complexity, such as its rank or number of distinct columns  are known \cite{razenshteyn2016weighted}. However, it would be very interesting to give polynomial time bicriteria approximation algorithms as we have done in the special case of binary $W$. 

\medskip
 
{\bf Acknowledgments:} David Woodruff would like to thank support from the 
Office of Naval Research (ONR) grant N00014-18-1-2562. 
Part of this work was done while the authors were visiting the 
Simons Institute for the Theory of Computing.
\bibliographystyle{alpha}
\bibliography{factor_analysis}


\appendix

\section{Alternate Approach via a Structural Result for Low-Rank Matrices}\label{app:alt}

In this appendix, we present a different approach to giving efficient algorithms for the masked low-rank approximation problem (Problem \ref{def:main}). Rather than considering the communication complexity of the mask matrix $W$, we prove a simple structural result about low-rank matrices: 
\begin{center}\emph{Any low-rank matrix cannot have too many ``heavy'' entries on its diagonal, or more generally, on the support of a column sparse matrix $W$.}\end{center}

We use this structural result to give an alternative proof of our main bicriteria approximation bound (Theorem \ref{thm:1}) in the Low-Rank Plus Diagonal (LRPD) and Low-Rank Plus Sparse (LRPS) setting. That is, when $W$ is either zero exactly on its diagonal, or only has a few zeros per column. In this setting, Theorem \ref{thm:1} yields Corollaries \ref{cor:lrpd} and \ref{cor:lrps} which show that  simply outputting a standard low-rank approximation of $A \circ W$ achieves error $OPT + \epsilon \norm{A}_F^2$ with rank $k' = O(kt/\epsilon)$, where $t$ is the maximum number of zeros in a column of $W$ ($t=1$ in the LRPD case). Our alternative proof applies to the same algorithm gives the same error bound with the same setting of $k'$. We are not aware of any formal connection between the two approaches or more generally the structural result and the communication complexity of $W$. Identifying such a connection would be very interesting.

We further show that the above structural result can be used  to obtain a fixed-parameter-tractable, relative error, non-bicriteria algorithm for 
Problem \ref{def:main} in the LRPD setting. We also give a fixed-parameter-tractable algorithm for the closely related Factor Analysis (FA) problem, where we want to decompose positive semidefinite $A$ as $L + D$ where $L$ is rank-$k$ and positive semidefinite and $D$ is diagonal and positive semidefinite.
Note that removing the positive semidefinite constraints, this is exactly equivalent to the LRPD problem. Minimizing $\norm{W \circ (A-L)}_F^2$ is equivalent to minimizing $\norm{A-(L+D)}_F^2$  where $D$ is diagonal, since given $L$ we can always set $D = \diag(A-L)$.

\smallskip
\spara{Application to Bicriteria Approximation:}

Consider Problem \ref{def:main} with mask matrix $W$ that has at most $t$ zeros per column (the LRPS approximation problem). 
As in the proof of Theorem \ref{thm:rand} (which yields Theorem \ref{thm:1} as a corollary), we will prove the bicriteria approximation bound via a comparison argument. Let $L_{opt}$ be any rank-$k$ matrix achieving error $OPT = \min_{\rank-k\ \hat L} \norm{W \circ (A-\hat L)}_F^2$. Since $L_{opt}$ is low-rank, by the structural result above (stated formally in Theorem \ref{thm:Mk2}), it cannot place significant mass on the entries in the sparse support of $1-W$, outside a small subset of rows (size $O(kt/\epsilon)$). This in turn implies the existence of a rank-$O(kt/\epsilon) + k$ matrix $\bar L$ that exactly matches $A \circ W$ on those rows and matches $L_{opt}$ on the rest of the matrix. Overall, $\bar L$ places very little weight outside the support of $W$. This is analogous to how $\bar L$ constructed in the proof of Theorem \ref{thm:rand} places no weight outside the support of the protocol communication matrix $W_\Pi$, which closely approximates $W$.

In Theorem \ref{thm:rand} we compare the error $\bar L$ to that of $L$ obtained by outputting a (near) optimal rank-$O(kt/\epsilon)$ approximation to $A \circ W$. Here we perform the same comparison. Since it is optimal, we have $\norm{A \circ W - L}_F^2 \le \norm{A\circ W - \bar L}_F^2$. On the entries outside the support of  $W$, $\bar L$ is already  very small and so close to $A \circ W$ (which is $0$ on these entries). Thus $L$ cannot give significantly smaller error than $\bar L$ on these entries. In turn, it cannot give significantly larger error on the entries in the support of $W$. This means that $L$ matches the approximation of $\bar L$, and in turn $L_{opt}$ on the support of $W$, yielding our bound. See Theorem \ref{thm:bicriteria} for a formal statement and proof.


\smallskip
\spara{Application to FPT Algorithm:}

In designing a fixed-parameter-tractable algorithm for the LRPD and FA problems we apply a recursive approach: we split our matrix into four quadrants and compute a low-rank plus diagonal decomposition of the top left and bottom right quadrants. 
Consider the case when $OPT$ is $0$: there is a rank-$k$ $L^\star$ with $\norm{W \circ (A-L^\star)}_F^2 = 0$. Equivalently, $A$ can be exactly decomposed as $A= D^\star + L^\star$ where $D^\star$ is diagonal. Note that this decomposition may not be unique. 
Letting $A_{11} = D^\star_{11}+L^\star_{11}$ denote the upper left quadrant, our recursively computed output $D'_{11}, L'_{11}$ satisfies:
$$D'_{11} + L'_{11} =A_{11} = D^\star_{11}+L^\star_{11}.$$
Thus $D'_{11}-D^\star_{11} = L^\star_{11} + L'_{11}$. Since this is a diagonal matrix and since it has rank $\le 2k$, we can see that it can have at most $2k$ nonzero entries. This is a special case of our main structural result (Theorem \ref{thm:Mk2}), which lets us make an analogous claim that $D'_{11}-D^\star_{11}$ does not have many large entries in the case when $A$ does not admit an exact decomposition. The same bound holds for the lower right quadrant and so overall, appending the recursively computed diagonal matrices, we have found $D^\star$ up to at most $4k$ incorrect entries.

If we iterate over all possible locations of these incorrect entries (a total of $O(n^{O(k)})$ possibilities), it only remains to solve the LRPD problem where we know all but $O(k)$ of the diagonal entries. This problem can be solved in $\poly(n)\cdot 2^{\poly(k)}$ time using generic polynomial solvers. Polynomial solvers have been used numerous times in the past to solve constrained low-rank approximation problems (see e.g., \cite{arora2012computing, moitra2012singly,razenshteyn2016weighted,basu2016computing}). Overall, since we need $n^{O(k)}$ guesses to succeed in identifying the incorrect entries, we obtain runtime $n^{O(k)} \cdot 2^{\poly(k)}$. 

\subsection{Additional Notation and Tools}

Throughout this section,
given an $n \times p$ matrix $M$ and $i \in [n], j \in [p]$ we let $M(i,j)$ denote its $(i,j)^{th}$ entry. For sets $R \subseteq [n]$, $C \subseteq [p]$ we let $M(R,C)$ be the $|R| \times |C|$ matrix composed of the intersection of the rows and columns indexed by $R$ and $C$ respectively. 
$\supp(M) \subseteq [n] \times [p]$ denotes the set of indices of $M$'s nonzero entries and $\nnz(M) = |\supp(M)|$ denotes the number of such entries. The above definitions all extend to vectors, except that we omit the index of the second dimension. So for a length-$n$ vector $m$, $m(i)$ is its $i^{th}$ entry and $\supp(m) \subseteq [n]$ is the set of indices of its nonzero entries.

We let $\orth(M)$ output $Q \in \R^{n \times \rank(M)}$ with orthonormal columns that span $M$.
For $M \in \R^{n \times n}$,  $M$ is positive semidefinite (PSD) if $x^T M x \ge 0$ for all $x\in \R^n$. We sometimes denote this by  $M \succeq 0$, with $M \succeq N$ denoting that $M-N \succeq 0$.
 Any $M \in \R^{n\times p}$ can be written via singular value decomposition as $M = U\Sigma V^T$ where $U \in \R^{n \times \rank(M)}, V\in \R^{p \times \rank(M)}$ have orthonormal columns and $\Sigma$ is a positive diagonal matrix containing the singular values of $M$, $\sigma_1(M) \ge ... \ge \sigma_{\rank(M)}(M) > 0$. We have $\sum_{i=1}^{\rank(M)} \sigma_i^2(M) = \norm{M}_F^2.$ The matrix pseudoinverse $M^+$ is given by  $M^+ = V \Sigma^{-1} U^T$, so that $M^+ M = VV^T$. 

\subsubsection*{Polynomial Solvers}
In our relative error approximation algorithms we make black-box use of polynomial system verifiers \cite{r92a,r92b,bpr96}, which can determine if there exists a solution to any given polynomial system of equations. We combine these verifiers with binary search techniques to perform polynomial optimization under polynomial constraints. Polynomial solvers have been used numerous times in the past to solve constrained low-rank approximation problems (see e.g., \cite{arora2012computing, moitra2012singly,razenshteyn2016weighted,basu2016computing}).

\begin{theorem} [Polynomial Decision Problem \cite{r92a,r92b,bpr96}]\label{thm:decision_solver}
Suppose we are given $m$ polynomial constraints over $v$ variables: $p_i (x_1, x_2, \ldots, x_v) \Delta_i 0$, where $\Delta_i$ is any of the ``standard relations'': $\{\geq, =, \leq \}$. Let $d$ denote the maximum degree of $p_i$ for $i\in [m]$ and let $H$ denote the maximum bit-length of the coefficients in all of the polynomial constraints. Then, in
\begin{equation*}
(m d)^{O(v)} \poly(H)
\end{equation*}
time, one can determine if there exist $x_1,\ldots,x_v$ satisfying all of the constraints. I.e., if $$\{x \in \R^v | \forall\, i\, p_i (x_1, x_2, \ldots, x_v) \Delta_i 0\} \neq \emptyset.$$
\end{theorem}

\subsection{Main Structural Result}

We begin with our main structural result.
Informally, if $W \in \{0,1\}^{n \times p}$ is a sparsity pattern with few zero entries in each column, then, ignoring a small subset of rows, no low-rank matrix can place significant mass outside the entries in $\supp(W)$ in comparison to those in this support. When $W = 1- I$ (corresponding to the LRPD problem), this implies that $L$ cannot place significant mass on all but a small subset of the diagonal. This special case generalizes a simple fact: if $L$ is low-rank and diagonal (i.e., has mass $0$ on the support of $W = 1-I$), then it has exactly $k$ non-zero diagonal entries.

\begin{theorem}[Main Structural Result]\label{thm:Mk2}
Consider any support matrix $W \in \{0,1\}^{n \times p}$ with at most $t$ zero entries per column.
For any rank-$k$ $L \in \R^{n \times p}$ and $0 \le \epsilon \le 1$, there is some subset of row indices $\mathcal{S} \subset  [n]$ with $|\mathcal{S}| \le \frac{tk}{\epsilon}$ such that, if we let $\bar W  = W([n]\setminus \mathcal{S},[p])$  and $\bar L = L([n]\setminus \mathcal{S},[p])$, then:

$$\norm{\bar L \circ (1-\bar W)}_F^2 \le \frac{\epsilon}{1-\epsilon} \cdot  \norm{L \circ W}_F^2.$$
\end{theorem}
That is, if we exclude $tk/\epsilon$ rows from $L$, the Frobenius norm mass outside the entries in $\supp(W)$ is at most an $\frac{\epsilon}{1-\epsilon}$ fraction of the mass on these entries. In the special case when $W = 1-I$, this gives that, excluding a subset of $O(k/\epsilon)$ entries, $L$ cannot have more than an $\epsilon$ fraction of its Frobenius norm mass on the diagonal.
\begin{proof}
Let $l_i$ denote the $i^{th}$ column of $L$. 
The $i^{th}$ row leverage score of $L$ is defined to be 
\begin{align}\label{eq:leverage}
\tau_i(L) = \max_{y \in \R^p} \frac{[Ly](i)^2}{\norm{Ly}_2^2}.
\end{align}
If we set $y  = e_j$ where $e_j$ is the $j^{th}$ standard basis vector in $\R^p$, we can see that $\tau_i(L)$ upper bounds how much mass the $i^{th}$ entry can have in column $j$, in comparison to the other entries of the column.
It is well known (see e.g. \cite{cohen2015uniform}) that when $L$ has rank $k$, the sum of leverage scores $\sum_{i=1}^n \tau_i(L) = k$. Thus we can see that there are at most $k/\epsilon$ rows with leverage score $\ge \epsilon$. In fact, it is possible to prove something even stronger: if we reweight at most $k/\epsilon$ rows of our matrix appropriately, then we can reduce all leverage scores to be simultaneously bounded by $\epsilon$. Formally:
\begin{claim}[Coherence Reducing Reweighting -- Lemma 1 of \cite{cohen2015uniform}]\label{clm:crr}
For any rank-$k$ $L \in \R^{n \times n}$ and $\beta > 0$, there exists a diagonal $D \in [0,1]^{n \times n}$ with at most $ k/\beta$ entries not equal to $1$ and 
$$\tau_i(DL) \le \beta\text{ for all }i \in [n].$$
\end{claim}
Armed with Claim \ref{clm:crr} we are ready to prove the theorem.
Let $D$ be the diagonal matrix guaranteed to exist by Claim \ref{clm:crr} with $\beta = \epsilon/t$. Let $\mathcal{S} = \{i: D(i,i) \neq 1\}$. We have $|\mathcal{S}| \le \frac{kt}{\epsilon}$. Additionally, note that for all $i \in [n]\setminus \mathcal{S}$ and all $j \in [p]$, $[DL](i,j) = L(i,j)$ since multiplying by $D$ does not reweight the $i^{th}$ row of $L$. Letting $z_i$ denote the $i^{th}$ column of $DL$, by  the leverage score bound of Claim \ref{clm:crr}, for all $i \in [n]$ and $j \in [p]$ we have:
\begin{align*}
\frac{[DL](i,j)^2}{\norm{z_i}_2^2} = \frac{[DLe_j](i)^2}{\norm{DLe_j}_2^2} \le \tau_i(DL) \le \frac{\epsilon}{t}.
\end{align*}
Thus, letting $\bar z_i$ denote the $i^{th}$ column of $DL \circ (1-W)$, which has at most $t$ nonzero entries:
\begin{align*}
\norm{\bar z_i}_2^2 \le \epsilon \cdot \norm{z_i}_2^2.
\end{align*}
Since $\bar z_i$ is just a subset of the entries in $z_i$, $\norm{z_i}_2^2 = (\norm{z_i-\bar z_i}_2^2) + \norm{\bar z_i}_2^2$ and so:
\begin{align*}
\norm{\bar z_i}_2^2 \le \frac{\epsilon}{1-\epsilon} \norm{z_i - \bar z_i}_2^2,
\end{align*}
where $z_i - \bar z_i$ is the $i^{th}$ row of $DL \circ W$.
This gives:
\begin{align*}
\norm{\bar L \circ (1-\bar W)}_F^2 \le \norm{DL \circ (1-W)}_F^2 \le \frac{\epsilon}{1-\epsilon} \norm{DL \circ W}_F^2 \le \frac{\epsilon}{1-\epsilon} \norm{L \circ W}_F^2,
\end{align*}
where we use that all entries of $D$ are in $[0,1]$. This completes the theorem.
\end{proof}

\subsection{Polynomial Time Bicriteria Approximation}\label{sec:bicriteria}

We start by using Theorem \ref{thm:Mk2} to give an alternative proof that performing a standard low-rank approximation of $A \circ W$ gives a strong bicriteria approximation bound for Problem \ref{def:main} when $W = 1-I$ or more generally has at most $t$ zeros per column (the LRPD and LRPS approximation problems). The bounds match those of Corollaries \ref{cor:lrpd} and \ref{cor:lrps} up to constants.

%
\begin{theorem}\label{thm:bicriteria} Let $W \in \{0,1\}^{n \times p}$ have at most $t$ zeros in each column. Then for $k' = \frac{6kt}{\epsilon}$ and $L$ with $\norm{A \circ W - L}_F^2 \le \min_{\rank-k'\ \hat L} \norm{A \circ W - \hat L}_F^2 + \epsilon_1 \norm{A}_F^2$, 
\begin{align*}
\norm{(A-L) \circ W}_F^2 \le OPT + \epsilon \norm{A}_F^2 + \epsilon_1 \norm{A}_F^2,
\end{align*}
where $OPT = \min_{\rank-k\ \hat L} \norm{(A-L)\circ W}_F^2$ is the optimal value of Problem \ref{def:main}.
\end{theorem}
\begin{proof}
Since $L$ minimizes $\norm{A \circ W -L}_F^2$ over all rank-$k'$ matrices up to a $\epsilon_1 \norm{A}_F^2$ additive factor, for any rank-$k'$ $\bar L$:
\begin{align*}
\norm{(A-L) \circ W}_F^2 \le \norm{A\circ W-L}_F^2 \le \norm{A\circ W-\bar L}_F^2  + \epsilon_1 \norm{A}_F^2.
\end{align*}
Thus to prove the theorem it suffices to exhibit  any rank-$\frac{6tk}{\epsilon}$ matrix $\bar L$ achieving
\begin{align*}
\norm{A \circ W-\bar  L}_F^2\le OPT + \epsilon \norm{A}_F^2. 
\end{align*}
Consider $L_{opt}$ achieving $OPT$. Let $\mathcal{S} \subseteq [n]$ be the set of rows given by applying Theorem \ref{thm:Mk2} to $L_{opt}$ with error parameter $\epsilon/5$. Let $\bar L$ be equal to $L_{opt}$ on the rows in $[n] \setminus \mathcal{S}$ and be equal to $A \circ W$ on the rows in $\mathcal{S}$. Note that $\bar L$ has rank $\le |\mathcal{S}| + k \le  \frac{6tk}{\epsilon}$. 
Since $\bar L$ exactly matches $A\circ W$ on the rows in $\mathcal{S}$ and matches $L_{opt}$ everywhere else
we can bound:
\begin{align}\label{thusHave}
\norm{A \circ W -\bar L }_F^2 &= \norm{[A \circ W - L_{opt}]([n]\setminus \mathcal{S},p)}_F^2\nonumber\\
&= OPT + \norm{[L_{opt} \circ (1- W)]([n]\setminus \mathcal S,p)}_F^2.
\end{align}
By  the guarantee of  Theorem \ref{thm:Mk2} applied with error parameter $\epsilon/5$ we have:
\begin{align}\label{thusHave2}
\norm{[L_{opt} \circ (1- W)]([n]\setminus \mathcal S,p)}_F^2 \le \frac{\epsilon/5}{1-\epsilon/5} \norm{L_{opt} \circ W}_F^2 \le  \frac{\epsilon}{4} \norm{L_{opt} \circ W}_F^2.
\end{align}
We can bound the right hand side since
we must have $\norm{L_{opt} \circ W}_F^2 \le 4 \norm{A}_F^2$.  Otherwise we would have by the triangle inequality:
\begin{align*}
\sqrt{OPT} = \norm{(A-L_{opt})\circ W}_F  \ge \norm{L_{opt} \circ W}_F - \norm{A \circ W}_F > \norm{A}_F,
\end{align*}
which is not possible as we can always set $L= 0$ and obtain objective function value $\norm{A}_F^2$. Plugging back into \eqref{thusHave} and \eqref{thusHave2} we thus have
 \begin{align*}
 \norm{A\circ W- \bar L}_F^2 \le OPT + \epsilon \norm{A}_F^2,
 \end{align*}
 which gives the theorem.
\end{proof}

\subsection{Fixed-Parameter-Tractable Algorithms}\label{sec:nk}

In this section we leverage Theorem \ref{thm:Mk2} in a different way: to give a fixed-parameter-tractable, relative error, non-bicriteria approximation algorithm for the LRPD approximation problem and the related constrained Factor Analysis (FA) problem using a simple recursive scheme. 

\subsubsection{Exact Decomposition}

For exposition, we first consider both problems in the case when $A$ can be exactly decomposed as $D^\star+L^\star$ where $D^\star$ is diagonal and $L^\star$ is rank $k$ (i.e., when $\min_{\rank-k\ L} \norm{W\circ(A-L)}_F^2 = 0$). In Section \ref{sec:appDec} we extend our techniques to solve the problems in full generality.
Our algorithm uses a very simple recursive approach: 
We split $A$ into four quadrants and compute LRPD (resp. Factor Analysis) decompositions of the upper left and lower right quadrants. Using Theorem \ref{thm:Mk2} we can prove that the diagonal matrices returned by these decompositions match $D^\star$ on all but $O(k)$ entries.\footnote{In general, the decomposition $A=D^\star+L^\star$ may not be unique, however this result holds for any $D^\star,L^\star$.} 

Letting $A_{11} = D^\star_{11}+L^\star_{11}$ denote the upper left quadrant, our algorithm recursively computes $D'_{11} + L'_{11} =A_{11} = D^\star_{11}+L^\star_{11}$. So $D'_{11}-D^\star_{11} = L^\star_{11} + L'_{11}$. Since this is a diagonal matrix and since it has rank $\le 2k$, by Theorem \ref{thm:Mk2} applied with $t =1$ and constant $\epsilon$, it can have at most $O(k)$ nonzero entries (since the norm of the remaining entries is bounded by the norm of the off-diagonal entries, which is $0$). Of course, Theorem \ref{thm:Mk2} is overkill here and we can see this fact directly. However, the more general theorem will be important in extending our result to the non-exact decomposition case. The same bound holds for the lower right quadrant and so overall, appending the recursively computed diagonal matrices, we have found $D^\star$ up to $O(k)$ incorrect entries.

If we guess the locations of these incorrect entries (which we can do with $O(n^{O(k)})$ guesses), it only remains to solve the LRPD problem where we know all but $O(k)$ of the diagonal entries. This problem can be solved in $\poly(n)\cdot 2^{\poly(k)}$ time using generic polynomial solvers. Overall, since we need $n^{O(k)}$ guesses to succeed in identifying the incorrect entries, we obtain runtime $n^{O(k)} \cdot 2^{\poly(k)}$. 

\begin{theorem}[Low-Rank Plus Diagonal Exact Decomposition]\label{thm:lrpd1}
There is an algorithm solving 
the LRPD problem (Problem \ref{def:main} with $W = 1-I$) up to additive error $1/2^{\poly(n)}$ in $n^{O(k)}\cdot 2^{O(k^2)}$ time when there exists and rank-$k$ $L^\star$ with $\norm{W \circ (A - L^\star)}_F^2 = 0$ and all entries of $A,L^\star$ are bounded in magnitude by $2^{\poly(n)}$. That is, the algorithm outputs rank-$k$ $L$ with: 
$$\norm{W \circ (A-L)}_F^2  \le \frac{1}{2^{\poly(n)}}.$$
\end{theorem}
Note that $\frac{1}{2^{\poly(n)}}$ additive error is on the order of error introduced by rounding an exact solution (with possibly irrational entries) to a $\poly(n)$ bit representation, so can be regarded as negligible. The assumption that $A,L$ have entries bounded by $2^{\poly(n)}$ is required to apply a polynomial solver is also very mild, as without this assumption these matrices could not in general be represented in $\poly(n)$ bits. 
\begin{proof}
As discussed, we use a recursive algorithm, cutting the size of $A$ in half in each step. In the base case, when $n \le k$ we can solve the problem trivially, simply by returning $L = A$. 

Fixing $L^\star$ and writing $D^\star = A - L^\star$, consider splitting $A$ into $4$ quadrants, each $n/2 \times n/2$:
\begin{align*}
\begin{bmatrix}
	A_{11} &A_{12}\\
	A_{21} &A_{22}
	\end{bmatrix} = \begin{bmatrix}
	D^\star_{11} &\\
	 &D^\star_{22}
	\end{bmatrix} + \begin{bmatrix}
	L^\star_{11} &L^\star_{12}\\
	L^\star_{21} &L^\star_{22}
	\end{bmatrix}.
\end{align*}
Since $A_{11} = D^\star_{11} + L^\star_{11}$ and  $A_{22} = D^\star_{22} + L^\star_{22}$, we know that Problem \ref{def:main} with $W = 1-I$ has optimum value $0$ on these submatrices.
Assume via induction that we compute rank-$k$ $L'_{11},L'_{22} \in  \R^{n/2 \times n/2}$ and diagonal $D'_{11},D'_{22} \in \R^{n/2 \times n/2}$ with 
such that $D'_{11}  = A_{11}- L'_{11}$ and $D'_{22}= A_{22} - L'_{22}.$\footnote{We actually only have these equalities up to $1/2^{\poly(n)}$ additive error. For now we ignore this detail. It will be formally handled in our proof for the general case, when $A$ is not exactly  equal to $D^\star+L^\star$ (Theorem \ref{thm:lrpd2}).}

By Theorem \ref{thm:Mk2}, $D'_{11}$ and $D^\star_{11}$ (similarly, $D'_{22}$ and $D^\star_{22}$) differ on at most $O(k)$ entries. Thus, letting $D' = \begin{bmatrix}
	D'_{11} &\\
	 &D'_{22}
	\end{bmatrix}$, $D'$ and $D^\star$ differ on at most $O(k)$ entries.
In $O(n^{O(k)})$ time we can iterate through every set of $O(k)$ indices $\mathcal{I} \subseteq [n]$. 
We can reorder the matrix so that the rows/columns corresponding to the indices in $\mathcal{I}$ are the first $O(k)$ rows/columns and repartition $A$ into $4$ quadrants:
\begin{align*}
\begin{matrix*}[r]
\scriptstyle{O(k)} \lbrace\\
\scriptstyle{n-O(k)} \lbrace
\end{matrix*}
\begin{bmatrix}
	A_{TT} &A_{TB}\\
	A_{BT} &A_{BB}
	\end{bmatrix} = \begin{bmatrix}
	D^\star_{TT} &\\
	 &D^\star_{BB}
	\end{bmatrix} + \begin{bmatrix}
	L^\star_{TT} &L^\star_{TB}\\
	L^\star_{BT} &L^\star_{BB}
	\end{bmatrix}.
\end{align*}
Here $A_{TT}$, $L^\star_{TT}$, and $D^\star_{TT}$ are $O(k) \times O(k)$ matrices. $A_{TB} = L^\star_{TB}$ is $O(k) \times (n-O(k))$. $A_{BT} = L^\star_{BT}$ is $(n-O(k)) \times O(k)$. And finally, $A_{BB}$, $L^\star_{BB}$, and $D^\star_{BB}$ are $(n-O(k)) \times (n-O(k))$. 

Assuming that $D'_{BB}$ matches  $D^\star_{BB}$ on all entries (which will happen for at least one guess of $O(k)$ indices $\mathcal{I}$), we can compute $L^\star_{BB}$ explicitly by setting $L^\star_{BB} = A_{BB} - D'_{BB}$. Let $U_B \in \R^{(n-O(k)) \times k}$ be an orthonormal span for the columns of $\begin{bmatrix}
	L^\star_{BT} &L^\star_{BB}\end{bmatrix}$ (which has rank $\le k$ since $L^\star$ is rank-$k$) and let $V_B\in \R^{n-O(k)\times k}$ be an orthonormal span for the rows of $\begin{bmatrix} L^\star_{TB}\\ L^\star_{BB}\end{bmatrix}$.
We know that $A = D^\star + L^\star$ can be written as:
\begin{align}
\label{eq:exact_fact}
\begin{bmatrix}
	A_{TT} &A_{TB}\\
	A_{BT} &A_{BB}
	\end{bmatrix} = \begin{bmatrix}
	D_{TT} &\\
	 &D'_{BB}
	\end{bmatrix} + \begin{bmatrix}
	U_T\\
	 U_B R_U
	\end{bmatrix} \begin{bmatrix}
	V_T^T & R_V^T V_B^T 
	\end{bmatrix},
\end{align}
where $R_V,R_U \in \R^{k\times k}$, $U_T,V_T \in \R^{O(k) \times k}$, and $D_{TT} \in \R^{O(k)\times O(k)}$ are unknown. In particular, one satisfying solution sets $D_{TT} = D_{TT}^\star$ and $R_U,R_V,U_T,V_T$ so that $\begin{bmatrix}
	U_T\\
	 U_B R_U
	\end{bmatrix} \begin{bmatrix}
	V_T^T & R_V^T V_B^T 
	\end{bmatrix} = L^\star$.

Equation \eqref{eq:exact_fact} is a degree-$2$ polynomial system in $O(k^2)$ unknown variables (the entries of $D_{TT},U_T,V_T,R_U,$ and $R_V$). Under the assumption that $A$ and $L^\star$ (and hence $D^\star = A - L^\star$) have entries bounded by $2^{\poly(n)}$, we can solve for entries satisfying \eqref{eq:exact_fact} up to $1/2^{\poly(n)}$ error using generic polynomial system solvers in $\poly(n) \cdot 2^{O(k^2)}$ time. We  detail this process in Section \ref{app:polySolver}. Our final runtime follows since at each layer of the recursion we require $n^{O(k)} \cdot 2^{O(k^2)}$ time, and since the problem size is cut in half in each layer, the total runtime is also bounded by $n^{O(k)} \cdot 2^{O(k^2)}$.
\end{proof}

We can give an algorithm for factor analysis in the exact decomposition case using very similar techniques to Theorem \ref{thm:lrpd1}.
\begin{theorem}[Factor Analysis Exact Decomposition]\label{thm:fa1}
There is an algorithm solving 
the Factor Analysis problem up to additive error $1/2^{\poly(n)}$ in $n^{O(k)}\cdot 2^{O(k^2 \log k)}$ time when $A = D^\star + L^\star$ for PSD $A,D^\star,L^\star$ and all entries of $A$ are bounded in magnitude by $2^{\poly(n)}$. That is, the algorithm outputs rank-$k$ $L \succeq 0$ and diagonal $D \succeq 0$ with: 
$$\norm{A-(D+L)}_F^2  \le \frac{1}{2^{\poly(n)}}.$$
\end{theorem}
\begin{proof}
Note that in the theorem statement we only assume the entries of $A$ are bounded by $2^{\poly(n)}$ and do not require an additional assumption on the entries of $L^\star$ as we did in Theorem \ref{thm:lrpd1}. This is simply because, since $D^\star,L^\star \succeq 0$, both have positive diagonal entries, and we should never choose $D^\star$ such that the sum of these entries is greater than the corresponding diagonal entry of $A$. Thus, the entries of $D^\star$ are bounded by $2^{\poly(n)}$. In turn, since $L^\star$ is an approximation of $A-D^\star$, which has bounded magnitude entries, $L^\star$ also has bounded magnitude entries (otherwise, a better approximation could be achieved by setting $L^\star=0$.)

We use the same recursive approach as Theorem \ref{thm:lrpd1}. The only difference is that we need to ensure that $D$ and $L$ are PSD. We already  have that $D'_{BB}$ is PSD (i.e., has all nonnegative entries) since it is composed of entries of $D'_{11}$ and $D'_{22}$, which are PSD by the recursive guarantee. When solving the final polynomial system we can ensure that $D_{TT}$ is PSD by adding $4k$ positivity constraints. Additionally since $A$ is PSD and symmetric, we have $U_B =V_B$. To ensure that $L$ is PSD we simple require that $R_U = R_V$ and that $U_T = V_T$. Thus, $L =  \begin{bmatrix}
	U_T\\
	 U_B R_U
	\end{bmatrix} \begin{bmatrix}
	U_T^T & R_U^T U_B^T \end{bmatrix}$, which is PSD. In solving the final polynomial system, this amounts to just using a single set of variables for the entries of $R_U = R_V$ and a set for the entries of $U_T = V_T$.

Overall, the runtime is similar to what is given in Theorem \ref{thm:lrpd1}, except the polynomial system now involves $O(k)$ constraints (each still with $O(1)$ degree and $O(k)$ variables) and thus requires $2^{O(k^2  \log k)} \cdot \poly(n)$ time to solve. See Section \ref{app:polySolver} for a detailed explanation.
\end{proof}

\subsubsection{Approximate Decomposition}\label{sec:appDec}
We now consider LRPD and FA problems when $\min_{D,L} \norm{A-(D+L)}_F^2 > 0$ and we wish to find $L,D$ achieving a $(1+\epsilon)$ approximation to this minimum.
Using Theorem \ref{thm:Mk2} we prove an analogous result to what we used in the exact decomposition case, namely, that given two near optimal low-rank plus diagonal approximations to $A$, their diagonal entries cannot differ outside a small subset of entries.
\begin{lemma}[Sparse Difference of Approximate Solutions]\label{lem:kdiffapp}
Consider any $A$ with $\norm{A-  (D + L)}_F^2 =  C$ and $\norm{A-(D'+L')}_F^2 \le C + \gamma$ for  diagonal $D,D'$,  rank-$k$ $L,L'$, and $\gamma \ge 0$.
Let $\mathcal{D}$ be either  the set of all diagonal matrices or of all nonnegative diagonal matrices. If $D = \argmin_{\hat D \in \mathcal{D}} \norm{A-(\hat D+L)}_F^2$ and $D' = \argmin_{\hat D \in \mathcal{D}} \norm{A-(\hat D+L')}_F^2$ then there is some set of $O(k/\epsilon^2)$ indices $\mathcal{S} \subseteq [n]$ with:
\begin{align*}
\sum_{i \notin \mathcal S} [D(i,i)-D'(i,i)]^2\le \epsilon^2 C + \epsilon^2 \gamma.
\end{align*}

\end{lemma}
For the LRPD problem, the set $\mathcal{D}$ will be the set of all diagonal matrices, in the factor analysis problem it will be all nonnegative diagonal matrices. If we solve one of these problems, fixing $L$, it is easy to compute $\argmin_{\hat D \in \mathcal{D}} \norm{A-(\hat D+L)}_F^2$, which can only improve our error. When $\mathcal{D}$ is all diagonal matrices, we just set $\hat D$ to $\diag(L-A)$. If $\mathcal{D}$ is all non-negative diagonal matrices we set $\hat D$ to $\max(0,\diag(L-A))$, where $\max$ denotes the entrywise maximum.
\begin{proof}
By triangle inequality we have:
\begin{align}\label{firsttri}
\norm{(D+L) - (D'+L')}_F^2 &\le \left (\sqrt{C} + \sqrt{C+\gamma}\right)^2\nonumber\\
&\le 4C + 2\gamma.
\end{align}
By  \eqref{firsttri} we also have (just ignoring any of the on-diagonal difference between $D+L$ and $D'+L$):
\begin{align*}
\norm{(L-L')-\diag(L-L')}_F^2 = \norm{(D+L)-(D'+L')-\diag((D+L)-(D'+L'))}_F^2 \le 4C + 2\gamma.
\end{align*}
Applying Theorem \ref{thm:Mk2} with $W = 1-I$ and $t=1$, since $L-L'$ has rank $\le 2k$, there is some set of $O(k/\epsilon^2)$ indices $\mathcal{S}$ such that:
\begin{align}\label{eq:Lcloseness}
\sum_{i  \in [n] \setminus \mathcal{S}} [L(i,i)-L'(i,i)]^2 = \sum_{i  \in [n] \setminus \mathcal{S}} [(A-L)-(A-L')](i,i)^2 \le \epsilon^2 C + \epsilon^2 \gamma.
\end{align}
If $\mathcal{D}$ is the set of all diagonal matrices, then 
by our assumption that $D$ is the closest diagonal matrix in $\mathcal{D}$ to $\diag(L-A)$ and $D'$ is the closest to $\diag(L'-A)$ we have $D(i,i) = -A(i,i)+L(i,i)$ and $D'(i,i) = -A(i,i)+L'(i,i)$. Combined with \eqref{eq:Lcloseness} this gives:
\begin{align*}
\sum_{i  \in [n] \setminus \mathcal{S}} [D(i,i)-D'(i,i)]^2 \le \epsilon^2 C + \epsilon^2 \gamma
\end{align*}
completing the lemma in this case.

Alternatively, if $\mathcal{D}$ is the set of all non-negative diagonal matrices, $D(i,i) = \max(0, -A(i,i)+L(i,i))$ and $D'(i,i) = \max(0,-A(i,i)+L'(i,i))$. If $A(i,i)-L(i,i)$ and $A(i,i)-L'(i,i)$ have the same sign, then either $D(i,i) = D'(i,i) = 0$ or $D(i,i) = -A(i,i)+L(i,i)$ and $D'(i,i) = -A(i,i)+L'(i,i)$. If they have opposite signs, then we have $|D(i,i) - D'(i,i)| \le |[A-L](i,i)-[A-L'](i,i)|$. Overall using \eqref{eq:Lcloseness} we again have:
\begin{align*}
\sum_{i  \in [n] \setminus \mathcal{S}} [D(i,i)-D'(i,i)]^2 \le \epsilon^2 C + \epsilon^2 \gamma,
\end{align*}
which completes the lemma.
\end{proof}

Using Lemma \ref{lem:kdiffapp} we can now give analogous results  to Theorems \ref{thm:lrpd1} and \ref{thm:fa2} using a similar recursive scheme.
As in the exact case, we will split $A$ into 4 quadrants and recursively approximate the top left and bottom right. This will yield $D'$ that nearly matches $D$ on all but $O(k/\epsilon^2)$ indices by Lemma \ref{lem:kdiffapp}. We can guess what these indices are in $n^{O(k/\epsilon^2)}$ time. We then subtract $D'$ from $A$ and have to solve LRPD/Factor Analysis with just $O(k/\epsilon^2)$ unknown entries. This last step is more complicated in the approximate case since we no longer know an exact low-rank span for the submatrix corresponding to the correct diagonals in $D'$. However, we can still solve the problem efficiently using ideas similar to the ``projection-cost preserving'' sketches of \cite{CohenElderMusco:2015}.

\begin{theorem}[Low-Rank Plus Diagonal Approximate Decomposition]\label{thm:lrpd2}
There an algorithm solving 
the LRPD problem (Problem \ref{def:main} with $W = 1-I$) to relative error $(1+\epsilon)$ and additive error $1/2^{\poly(n)}$ in $n^{O(k/\epsilon^2)}\cdot 2^{O(k^2/\epsilon^2)}$ time, assuming that all entries of $A$ are bounded in magnitude by $2^{\poly(n)}$ and there exists an optimal $L^\star$ with all entries bounded in magnitude by $2^{\poly(n)}$. That is, the algorithm outputs rank-$k$ $L$ with: 
$$\norm{W \circ (A-L)}_F^2  \le  (1+\epsilon)  \norm{W \circ ({A} -  L^\star)}_F^2 +\frac{1}{2^{\poly(n)}}.$$
\end{theorem}
Note that Theorem \ref{thm:lrpd1} follows as a special case, when $\norm{W \circ (A-L^\star)}_F^2 = 0$ and we set $\epsilon = \Omega(1)$.
\begin{proof}
As discussed, we use a recursive algorithm, cutting the size of $A$ in half in each step. In the base case, when $n \le k$ we can solve the problem trivially, simply by returning $L = A$. 
Fix $L^\star$ and $D^\star = \diag(A-L^\star)$ with
\begin{align*}
L^\star \in \argmin_{\text{rank-$k$} \hat L} \norm{W \circ ({A} -  \hat L)}_F^2.
\end{align*}
and let $\Delta^\star = A-(D^\star+L^\star)$.
Split $A$ into $4$ quadrants, each $n/2 \times n/2$:
\begin{align*}
\begin{bmatrix}
	A_{11} &A_{12}\\
	A_{21} &A_{22}
	\end{bmatrix} = \begin{bmatrix}
	D^\star_{11} &\\
	 &D^\star_{22}
	\end{bmatrix} + \begin{bmatrix}
	L^\star_{11} &L_{12}\\
	L^\star_{21} &L_{22}
	\end{bmatrix} + \begin{bmatrix}
	\Delta^\star_{11} &\Delta^\star_{12}\\
	\Delta^\star_{21} &\Delta^\star_{22}
	\end{bmatrix}.
\end{align*}
Assume that we recursively compute rank-$k$ $L'_{11},L'_{22} \in  \R^{n/2 \times n/2}$ and diagonal $D'_{11},D'_{22} \in \R^{n/2 \times n/2}$ with $D'_{ii} = \diag(A-L'_{ii})$ such that for $i \in \{1,2\}$:
\begin{align*}
\norm{A_{ii} - (D'_{ii} + L'_{ii})}_F^2 \le (1+\epsilon) \min_{\substack{\text{rank-$k$ } \hat L\\ \text{ diagonal } {\hat D}}} \norm{{A}_{ii} - ({\hat D} + \hat L)}_F^2 + \frac{1}{2^{\poly(n)}}.
\end{align*}
Note that this ensures in particular that for both $i \in \{1,2\}$, 
$$\norm{A_{ii} - (D'_{ii} + L'_{ii})}_F^2 \le (1+\epsilon)\norm{\Delta^\star_{ii}}_F^2 + \frac{1}{2^{\poly(n)}}.$$ 
Applying
Lemma \ref{lem:kdiffapp}, we then have that there is some set of $\mathcal{S}_i$ indices with $ |\mathcal{S}_i| = O(k/\epsilon^2)$  such that, letting $\bar D_{ii}$ match $D^\star_{ii}$ on these indices and match $D'_{ii}$ everywhere else, 
$$\norm{D^\star_{ii} - \bar D_{ii}}_F^2 \le \epsilon^2 \left ( \norm{\Delta^\star_{ii}}_F^2 + \frac{1}{2^{\poly(n)}} \right ).$$
Thus, letting $\bar D = \begin{bmatrix}
	\bar D_{11} &\\
	 & \bar D_{22}
	\end{bmatrix}$, we have
	
	\begin{align}\label{eq:barMatch}
	\norm{D^\star - \bar D}_F^2 \le \epsilon^2 \left (\norm{\Delta^\star_{11}}_F^2 + \norm{\Delta^\star_{22}}_F^2\right ) + \frac{1}{2^{\poly(n)}} \le \epsilon^2 \norm{\Delta^\star}_F^2 + \frac{1}{2^{\poly(n)}}.
	\end{align}
	
	After computing $D'_{11}$ and $D'_{22}$ recursively, we know $\bar D$ up to $O(k/\epsilon^2)$ entries. We know the entries matching those of $D'$ and do not know those in $\mathcal{S}_1 \cup \mathcal{S}_2$ matching $D^\star$. In $n^{O(k/\epsilon^2)}$ time we can iterate through every set of $O(k/\epsilon^2)$ indices $\mathcal{I} \subseteq [n]$. 
We can order the matrix so that the rows/columns corresponding to the indices in $\mathcal{I}$ are the first $O(k/\epsilon^2)$ rows/columns and repartition $A$ into $4$ quadrants:
\begin{align*}
\begin{matrix*}[r]
\scriptstyle{O(k/\epsilon^2)} \lbrace\\
\scriptstyle{n-O(k/\epsilon^2)} \lbrace
\end{matrix*}
\begin{bmatrix}
	A_{TT} &A_{TB}\\
	A_{BT} &A_{BB}
	\end{bmatrix} = \begin{bmatrix}
	D^\star_{TT} &\\
	 &D^\star_{BB}
	\end{bmatrix} + \begin{bmatrix}
	L^\star_{TT} &L^\star_{TB}\\
	L^\star_{BT} &L^\star_{BB} 
	\end{bmatrix} + \begin{bmatrix}
	\Delta^\star_{TT} &\Delta^\star_{TB}\\
	\Delta^\star_{BT} &\Delta^\star_{BB} 
	\end{bmatrix}.
\end{align*}

Assume that $\bar D_{BB} = D'_{BB}$, which will happen for at least one guess of $O(k/\epsilon^2)$ indices $\mathcal{I}$.
By the triangle inequality and \eqref{eq:barMatch} we have:
\begin{align}
 \min_{\substack{\text{rank-$k$ }  L \\ {\text{diagonal }   D_{TT}} }} &\left \|{A} - \left(\begin{bmatrix}  D_{TT} &\\ & D'_{BB}\end{bmatrix} +  L\right)\right\|_F^2 \label{eq:min_problem_old}\\
&\le \norm{A-(\bar D + L)}_F^2 \nonumber\\ 
&\le (\norm{A-(D^\star+L)}_F + \norm{D^\star-\bar D}_F)^2\nonumber\\
&\le (1+\epsilon)^2 \norm{\Delta^\star}_F^2 + \frac{1}{2^{\poly(n)}} \label{eq:insteadOpt1}
\end{align}
where we make use of the fact that certainly $\norm{\Delta^\star}_F^2 \le \norm{A}_F^2 \le 2^{\poly(n)}$ by the assumption that $A$ has bounded entries.
Thus to approximately solve LRPD, it suffices to solve the minimization problem in \eqref{eq:min_problem_old}. Since ${L}$ is rank $k$, we can rewrite this problem as:
\begin{align}
	\label{eq:min_problem}
	 \min_{\substack{Z_T,W_T \in \R^{O(k/\epsilon^2) \times k} \\ Z_B, W_B \in \R^{n-O(k/\epsilon^2) \times k} \\  {\text{diagonal }   D_{TT}} }} &\bnorm{{A} - \left(\begin{bmatrix}  D_{TT} &\\ & D'_{BB}\end{bmatrix} +  \begin{bmatrix}Z_T \\ Z_B \end{bmatrix} \begin{bmatrix}W_T^T & W_B^T \end{bmatrix}\right)}_F^2.
\end{align}

Unfortunately, we cannot efficiently solve this minimization problem directly because it involves $O(nk)$ unknowns: $Z_T$, $Z_B$, $W_T$, and $W_B$ contain $2nk$ free variables in total, so applying Theorem \ref{thm:decision_solver} would take time exponential in $n$.

Recall that for the exact decomposition problem, we reduced this cost by writing ${L}$ exactly in a factored form, using the rank-$k$ column and row spans of $[L^\star_{BT},L^\star_{BB}]$ and $[L^\star_{TB};L^\star_{BB}]$. This reduced the number of free variables to $O(k^2)$ -- see \eqref{eq:exact_fact}. Because $\begin{bmatrix}A_{BT} & A_{BB}-D_{BB}' \end{bmatrix}$ and $\begin{bmatrix}A_{TB} \\ A_{BB} -D_{BB}'\end{bmatrix}$ do not in general exactly match $L^\star$ and so have rank $k$ in the approximate decomposition problem, this is no longer possible. However, it \emph{is possible} to show that \eqref{eq:min_problem} can be solved approximately with an ${L}$ that is restricted to a particular factored form involving few variables.
Specifically, let:
\begin{align*}
U_B \in \R^{(n - O(k/\epsilon^2)) \times \lceil k/\epsilon \rceil} \text{ contain the top $\lceil k/\epsilon \rceil$ column singular vectors of } \begin{bmatrix}A_{BT} & A_{BB} - D_{BB}' \end{bmatrix} ,
\end{align*}
 \begin{align*}
 V_B \in \R^{(n - O(k/\epsilon^2)) \times \lceil k/\epsilon \rceil} \text{ contain the top $\lceil k/\epsilon \rceil$ row singular vectors of } \begin{bmatrix}A_{TB} \\ A_{BB}-D_{BB}' \end{bmatrix}.
 \end{align*}
%
%
\begin{claim}\label{clm:approxSpan}
For $U_B$ and $V_B$ as described above, 
\begin{align}\label{eq:insteadOpt}
 &\min_{\substack{ R_U,  R_V \in \R^{O(k/\epsilon^2) \times k}\\
							 Z_T, W_T \in \R^{O(k/\epsilon^2) \times k}\\
							 \text{diagonal }  D_{TT}}} 
							 \left \|A -\left ( \begin{bmatrix}  D_{TT} &\\ & D'_{BB}\end{bmatrix} + \begin{bmatrix}Z_T \\ U_B R_U \end{bmatrix} \begin{bmatrix}W_T^T & R_V^T  V_B^T \end{bmatrix}\right )\right \|_F^2 \leq \\
							 &(1 + \epsilon)^2\min_{\substack{Z_T,W_T \in \R^{O(k/\epsilon^2) \times k}\\ Z_B, W_B \in \R^{n-O(k/\epsilon^2) \times k}  \\  {\text{diagonal }   D_{TT}} }} \bnorm{{A} - \left(\begin{bmatrix}  D_{TT} &\\ & D'_{BB}\end{bmatrix} +  \begin{bmatrix}Z_T \\ Z_B \end{bmatrix} \begin{bmatrix}W_T^T & W_B^T \end{bmatrix}\right)}_F^2 \nonumber.
\end{align}
\end{claim}
It immediately follows from \eqref{eq:insteadOpt1} and the equivalence of \eqref{eq:min_problem_old} and \eqref{eq:min_problem} that, if we solve \eqref{eq:insteadOpt}, we will obtain an LRPD approximation with error $\leq (1+\epsilon)^3\|\Delta^\star\|_F^2 + \frac{1}{2^{\poly(n)}}$. This can be done using a generic polynomial solver because, in contrast to \eqref{eq:min_problem}, \eqref{eq:insteadOpt}  has just $O(k^2/\epsilon^2)$ unknowns, and can be turned into a polynomial system with degree $O(1)$ and $O(1)$ constraints. Thus it can be solved to  $\frac{1}{2^{\poly(n)}}$ error in $\poly(n) \cdot 2^{O(k^2/\epsilon^2)}$. We describe this process in detail in Section \ref{app:polySolver}. Our final runtime follows since at each layer of recursion we require $n^{O(k/\epsilon^2)} \cdot 2^{O(k^2/\epsilon^2)}$ time and since the problem size is cut in half at each layer (recall that the $n^{O(k/\epsilon^2)}$ term comes from the fact that we must guess which entries of $D'$ are close to those of $D^\star$).
Thus, to complete the proof of Theorem \ref{thm:lrpd2}, it just remains to prove Claim \ref{clm:approxSpan}.

\begin{proof}[Proof of \ref{clm:approxSpan}]
	Let $Z_T^\star$, $Z_B^\star$, $W_T^\star$, $W_B^{\star}$, and $D_{TT}^\star$ comprise an optimal solution for \eqref{eq:min_problem} and denote:
	\begin{align*}
	\bar D^\star = \begin{bmatrix}D_{TT}^\star & \\
	& D'_{BB}\end{bmatrix}
	\text{ and }
	\bar L^\star = \begin{bmatrix}Z_T^\star \\ Z_B^\star \end{bmatrix} \begin{bmatrix}W_T^{\star T}  & W_B^{\star T} \end{bmatrix} 
	\end{align*}
	We will make the argument in two steps. First we consider a third minimization problem:
	\begin{align}\label{eq:minbetween}
	&\min_{\substack{ R_U \in \R^{O(k/\epsilon^2) \times k}\\
			Z_T, W_T \in \R^{O(k/\epsilon^2) \times k}\\
			W_B \in \R^{n-O(k/\epsilon^2) \times k} \\
			\text{diagonal }  D_{TT}}} 
			\left \|A -\left ( \begin{bmatrix}  D_{TT} &\\ & D'_{BB}\end{bmatrix} + \begin{bmatrix}Z_T \\ U_B R_U \end{bmatrix} \begin{bmatrix}W_T^T & W_B^T \end{bmatrix}\right )\right \|_F^2.
	\end{align}
	Let $OPT_1$ be the optimum value of \eqref{eq:min_problem}, $OPT_2$ be the optimum value of \eqref{eq:minbetween}, and $OPT_3$ be the optimum value of  \eqref{eq:insteadOpt}. We will establish the claim by separately showing:
	\begin{align}
	\label{eq:project1}
	OPT_2 &\leq (1+\epsilon)OPT_1 \\
	\label{eq:project2}
	OPT_3 &\leq (1+\epsilon)OPT_2.
	\end{align}
	We prove \eqref{eq:project1} first. 
	With $Z_T^\star$, $Z_B^\star$, $W_T^\star$, and $W_B^\star$ as defined above, let $\tilde{L}_1$ equal:
	\begin{align*}
	\tilde{L}_1 = \begin{bmatrix}Z_T^\star \\ \left(U_BU_B^T\right)Z^\star_B \end{bmatrix} \begin{bmatrix}W_T^{\star T}  & W_B^{\star T} \end{bmatrix}.
	\end{align*}
	Since $\|A - (\bar D^\star + \tilde{L}_1)\|_F^2 \geq OPT_2$, if we can show that $\|A - (\bar D^\star + \tilde{L}_1)\|_F^2\leq (1+\epsilon)OPT_1$, we prove \eqref{eq:project1}.
	To do so, we introduce the additional notation:
	\begin{align*}
	F &= \begin{bmatrix} A_{BT} & A_{BB}- D'_{BB}\end{bmatrix} & \text{ and } & &Q &= \orth\left(\begin{bmatrix} W^\star_T \\ W^\star_B\end{bmatrix}\right).
	\end{align*}
	We first notice that we always have $Z^\star_B\begin{bmatrix}W_T^{\star T} & W_B^{\star T} \end{bmatrix} = FQQ^{T}$. This holds because, in choosing the minimal solution to \eqref{eq:min_problem}, after fixing all variables besides $Z_B$, we can always obtain a better solution by choosing $Z_B$ so that $Z_B\begin{bmatrix}W_T^T & W_B^T \end{bmatrix}$ is the optimal approximation to $F$ in the row span of $\begin{bmatrix}W_T^T & W_B^T \end{bmatrix}$. This is obtained by setting $Z_B = \begin{bmatrix}W_T^T & W_B^T \end{bmatrix}^+$, in which case $Z_B = FQQ^T$ for $Q = \orth\left(\begin{bmatrix}W_T^T & W_B^T \end{bmatrix}\right)$.
	
	Then, since $\bar  L^\star$ and $\tilde{L}_1$ only differ on the lower block of their left factor, we can see that:
	\begin{align}
	\label{eq:simp_diff}
	\norm{A - (\bar D^\star + \tilde{L}_1)}_F^2 - \norm{A - (\bar D^\star + \bar L^\star)}_F^2  = \|F - FQQ^T\|_F^2 - \|F - U_BU_B^TFQQ^T\|_F^2. 
	\end{align}
	We have:
	\begin{align*}
		\|F - (U_BU_B^T)FQQ^T\|_F^2 &= \|F(I - QQ^T) + (F - U_BU_B^TF)QQ^T\|_F^2 \\
		&= \|F(I - QQ^T)\|_F^2 + \|(F - U_BU_B^TF)QQ^T\|_F^2. 
	\end{align*}
	So plugging into \eqref{eq:simp_diff}, we have
	\begin{align}
		\label{eq:simp_diff2}
		\norm{A - (\bar D^\star + \tilde{L}_1)}_F^2 - \norm{A - (\bar D^\star + \bar L^\star)}_F^2   =  \|(F - U_BU_B^TF)QQ^T\|_F^2. 
	\end{align}
	We just need to show that $\|(F - U_BU_B^TF)QQ^T\|_F^2$ is small. 
	
	To do so, let $\sigma_1, \ldots, \sigma_n$ denote the singular values of $F$. Since $U_{B}U_B^T$ was chosen to be the top $\lceil k/\epsilon\rceil$ singular vectors of $F$, $F - U_BU_B^TF$ has singular values:
	\begin{align*}
	\sigma_{\lceil k/\epsilon\rceil+1}, \sigma_{\lceil k/\epsilon\rceil + 2},\ldots, \sigma_{n}.
	\end{align*}
	Since $Q$ has orthonormal columns, $QQ^T$ is a rank $k$ projection matrix and thus:
	\begin{align*}
	\|(F - U_BU_B^TF)QQ^T\|_F^2 \leq \sum_{i=\lceil k/\epsilon\rceil+1}^{\lceil k/\epsilon\rceil + k} \sigma_i^2 \leq \epsilon \sum_{i=1}^{\lceil k/\epsilon\rceil + k} \sigma_i^2 \leq \epsilon \|F - F_k\|_F^2,
	\end{align*}
	where $F_k$ is the optimal rank $k$ approximation to $F$ in Frobenius norm. Note that this argument is essentially identical to the proof in \cite{CohenElderMusco:2015} that the top $\lceil k/\epsilon\rceil$ singular vectors of a matrix can be used to form a ``projection-cost preserving sketch'' of that matrix. 
	
	Finally, it of course holds that $\|F - F_k\|_F^2 \leq OPT_1$ since $OPT_1 \leq \|F - Z^\star_B\begin{bmatrix}W_T^{\star T} & W_B^{\star T} \end{bmatrix} = \|F - FQQ^T\|_F^2$. Since $Q$ is rank $k$, this is a rank $k$ approximation to $F$, so $\|F - F_k\|_F^2 \leq \|F - FQQ^T\|_F^2$.
	This concludes the proof of \eqref{eq:project1}: returning to \eqref{eq:simp_diff2}, we see that $\norm{A - (\bar D^\star + \tilde{L}_1)}_F^2 \leq (1+\epsilon)OPT_1$ and it follows that $OPT_2 \leq \norm{A - (\bar D^\star + \tilde{L}_1)}_F^2  \leq (1+\epsilon)OPT_1$.
	
	The proof of \eqref{eq:project2} is essentially identical, so we omit it for brevity. Briefly, if we let $Z_T^\star$, $R_U^\star$, $W_T^\star$, $W_B^{\star}$, and $D_{TT}^\star$  comprise an optimal solution for \eqref{eq:minbetween} and $\bar D^\star,\bar L^\star$ be defined as before. Denote:
	\begin{align*}
		F &= \begin{bmatrix} A_{TB} \\ A_{BB}- D'_{BB}\end{bmatrix} &\text{ and } & & Q &= \orth\left(\begin{bmatrix} Z^\star_T \\ U_B R_U^\star \end{bmatrix}\right)
	\end{align*}
	and then we have that $\begin{bmatrix} Z^\star_T \\ U_B R_U^\star \end{bmatrix}W_B^{\star} = QQ^TF$. This is all we need to prove that 
	\begin{align*}
		\tilde{L}_2 = \begin{bmatrix}Z_T^\star \\ U_BR_U^\star\end{bmatrix} \begin{bmatrix}W_T^{\star T}  & W_B^{\star T}\left(V_BV_B^T\right) \end{bmatrix},
	\end{align*}
	satisfies $\norm{A - (\bar D^\star + \tilde{L}_2)}_F^2 \leq (1+\epsilon)OPT_2$. We also have $\norm{A - (\bar D^\star + \tilde{L}_2)}_F^2 \geq OPT_3$, which yields \eqref{eq:project2}, and thus the proof of Claim \ref{clm:approxSpan}.
\end{proof}

\end{proof}

We have an analogous result to Theorem \ref{thm:lrpd2} for the factor analysis problem, which is also proved in essentially the same way. As in the exact decomposition case, we just need to require that in the final optimization problem, $D_{TT}$ is positive, $Z_T = W_T$ and $R_U = R_V$. Overall we have: 
\begin{theorem}[Factor Analysis Approximate Decomposition]\label{thm:fa2}
There is an algorithm solving 
the Factor Analysis problem to relative error $(1+\epsilon)$ and additive error $1/2^{\poly(n)}$ in $n^{O(k/\epsilon^2)}\cdot 2^{O(k^2\log k/\epsilon^2)}$ time, assuming that all entries of $A$ are bounded in magnitude by $2^{\poly(n)}$. That is, the algorithm outputs rank-$k$ $L \succeq 0$ and diagonal $D \succeq 0$ with: 
$$\norm{A-(D+L)}_F^2  \le  (1+\epsilon)  \norm{{A} - ({ D^\star} +  L^\star)}_F^2 +\frac{1}{2^{\poly(n)}}.$$
\end{theorem}
Note that Theorem \ref{thm:fa1} follows as a special case when $\norm{{A} - ({ D^\star} +  L^\star)}_F^2 = 0$ and we set $\epsilon = \Omega(1)$.

\subsection{Polynomial Optimization}\label{app:polySolver}

In this section we detail our use of polynomial solvers in Theorems \ref{thm:lrpd1}, \ref{thm:fa1}, \ref{thm:lrpd2}, and \ref{thm:fa2}. We start with a technical lemma that we will use to bound the bit complexity of certain polynomial systems that arise in our algorithms.

\begin{lemma}[Bounded Low-Rank Factors]\label{lem:factorBound} Consider $M \in \R^{n \times p}$ with $|M(i,j)|\le \Delta$ for all $i,j$ and let $L\in \R^{n \times p}$ be an optimal rank-$k$ approximation:
\begin{align*}
L = \argmin_{\text{rank-$k$ } L} \norm{M-L}_F.
\end{align*}
We can write $L = UV^T$ for $U \in \R^{n \times k}$, $V \in \R^{p \times k}$ with $|U(i,j)| \le \Delta$ and $|V(i,j)| \le n^{1/4}\Delta^{1/2}$ for all $i,j$.
\end{lemma}
\begin{proof}
Assume that $\Delta = 1$. For general $\Delta$, we can first scale $M$ by $1/\Delta$ so that its entries are bounded in magnitude by $1$, exhibit $U,V$ with entries bounded in magnitude by $n^{1/4}$, and then scale these matrices each up by a $\sqrt{\Delta}$ factor, yielding the result.

We can decompose $L = UV^T$ using e.g., a QR decomposition, where $U$ is orthonormal and so has entries bounded in magnitude by $1$. Consider the $i^{th}$ row of $V$, $v_i \in \R^k$. If we project this row to the row span of $U$, $UV^T$ will be unchanged and so will still equal $L$. Additionally, after this projection, we can see that, letting $l_i,m_i \in \R^p$ be the $i^{th}$ columns of $L$ and $M$:  
\begin{align}\label{eq:vBound}
\norm{v_i}_2  = \norm{Uv_i}_2 = \norm{l_i}_2 \le \norm{m_i}_2 \le \sqrt{n}.
\end{align}
where the second to last inequality follows from the fact that if $\norm{l_i}_2 > \norm{m_i}_2$, we could achieve a better low-rank approximation by replacing $\norm{l_i}_2$ with the projection of $m_i$ onto this column, which has norm $\le \norm{m_i}_2$ and just rescales the column so does not change the rank of $L$.

By \eqref{eq:vBound}, $V$ has all entries bounded in magnitude by $\sqrt{n}$. Finally, if we multiply $U$ by $n^{1/4}$ and divide $V$ by $n^{1/4}$ we will still have $UV^T = L$ and all entries will be bounded in magnitude by  $n^{1/4}$, completing the lemma.
\end{proof}

We now discuss our use of polynomial solvers in the exact decomposition case, considered in Theorems \ref{thm:lrpd1} and \ref{thm:fa1}.

\subsubsection{Use of Polynomial Solvers in Theorems \ref{thm:lrpd1} and \ref{thm:fa1}}

In Theorem \ref{thm:lrpd1} we reduce solving the LRPD decomposition problem to finding matrices $R_V,R_U \in \R^{k\times k}$, $U_T,V_T \in \R^{4k \times k}$, and $D_{TT} \in \R^{4k\times 4k}$ satisfying \eqref{eq:exact_fact}. This equation defines a polynomial system of equations. This system has $O(k^2)$ variables: the entries of $R_V,R_U \in \R^{k\times k}$, $U_T,V_T \in \R^{4k \times k}$, and $D_{TT} \in \R^{4k\times 4k}$. It has $n^2$ constraints, each of the form 
$$A(i,j) = \wh D(i,j) +  \wh L(i,j)$$
 where $ \wh D =  \begin{bmatrix}
	D_{TT} &\\
	 &D'_{BB}
	\end{bmatrix}$ and $ \wh L = \begin{bmatrix}
	U_T\\
	 U_B R_U
	\end{bmatrix} \begin{bmatrix}
	V_T^T & R_V^T V_B^T 
	\end{bmatrix}.$
	Note that $A(i,j)$ is a known quantity and $ \wh D(i,j) +  \wh L(i,j)$ is a degree-$2$ polynomial in the unknown variables. We can combine all these equality constraints into a single constraint:
	\begin{align}\label{eq:mainConstraint}
	\sum_{i,j} \left (A(i,j) = \wh D(i,j) + \wh L(i,j)\right )^2 = 0.
	\end{align}
	Note that the left hand side of this constraint is a degree-$4$ polynomial in the unknown variables. 
	We will identify entries of $R_V,R_U,U_T,V_T$ and $D_{TT}$ that satisfy \eqref{eq:mainConstraint} up to additive error $1/2^{\poly(n)}$ via binary  search. To apply  binary search we must first bound the range that we must search over. 
	
	\subsubsection{Bounding the Range of Binary Search}
	
	By the assumption of Theorem \ref{thm:lrpd1} that $A = D^\star + L^\star$ where $D^\star$ and $L^\star$ have entries bounded in magnitude by $2^{\poly(n)}$, there is a solution to \eqref{eq:mainConstraint} where the entries of $\wh D$ (and hence $D_{TT}$) and $\wh L$ are bounded in magnitude by $2^{\poly(n)}$. By  Lemma \ref{lem:factorBound}, if this is the case, $\wh L$ admits a factorization with all entries bounded in magnitude by $2^{\poly(n)}$. Thus there is a solution with all entries of $U_T,V_T, U_BR_U, V_BR_V$ bounded in magnitude by $2^{\poly(n)}$. Finally, we note that this implies the existence of a solution with all entries of $R_U,R_V$ bounded in magnitude by  $2^{\poly(n)}$: we can assume that the columns of $R_U$ fall in the row span of $U_B$ (resp. for $R_V$ and $V_B$), thus their norms and hence entries are bounded by the column norms of $U_BR_U$, which are bounded by $2^{\poly(n)}$.  Thus, for all variables, we can restrict our search range to $[-2^{\poly(n)},2^{\poly(n)}]$.
	 
	 \subsubsection{Performing an Iteration of Binary  Search}
	 
	 To perform binary search we will consider the polynomial system consisting of \eqref{eq:mainConstraint}, along with a constraint that all variables are bounded in magnitude by $2^{\poly(n)}$ (i.e., a degree-$2$ constraint that bounds the sum of squares of the variables by $2^{\poly(n)}$).
	 In each iteration of binary search, we will consider an entry $M(i,j)$ (where $M$ is $R_V,R_U,U_T,V_T$, or $D_{TT}$) and verify this polynomial system augmented with the search constraints $c_1 \le M(i,j)$ and $M(i,j) \le c_2$ for any $c_1,c_2 \in [-2^{\poly(n)},2^{\poly(n)}]$. The overall system has four constraints that are $O(1)$ degree polynomials in $O(k^2)$ variables. 
	Its coefficients are all multiples of: 
	\begin{enumerate}
	\item The entries of $D'_{BB}$, which are obtained via a recursive call to our algorithm and thus bounded in magnitude by $2^{\poly(n)}$.
	\item The entries of $U_B$ and $V_B$ which are bounded in magnitude by $1$ since they are orthonormal.
	\item The entries of $A$ and the values $c_1$, and $c_2$, which again are bounded in magnitude by $2^{\poly(n)}$. 
	\end{enumerate}
	Thus, as long as we round all coefficients to additive error $1/2^{\poly(n)}$, we can represent our full system using $H = \poly(n)$ bits. Since, by our added magnitude constraint, any valid solution to the system has all variables bounded by  $2^{\poly(n)}$, 
	rounding these coefficients can only affect the minimum value of \eqref{eq:mainConstraint} by $1/2^{\poly(n)}$. We can verify that this minimum value is $\le 1/2^{\poly(n)}$  in $2^{O(k^2)} \cdot \poly(n)$ time using Theorem \ref{thm:decision_solver}. This certifies that there is some setting of $M(i,j) \in [c_1,c_2]$ that gives a solution with the left hand side of \eqref{eq:mainConstraint} bounded by  $ 1/2^{\poly(n)}$. By performing $\poly(n)$ iterations in this way, we can identify $c_1,c_2$ that are $1/2^{\poly(n)}$ apart such that some value of  $M(i,j) \in [c_1,c_2]$ achieves such a value.
	
	\subsubsection{Fixing Variables}
	
	After identifying a small range $[c_1,c_2]$ in which some valid value for $M(i,j)$ lies, we will fix this variable. Specifically, we will chose an arbitrary value $c^*$ in this range and add the constraint $M(i,j) = c^*$ to our system by adding the term $(M(i,j)-c^*)^2$ to the left hand side of \eqref{eq:mainConstraint}. $c^*$ is within $1/2^{\poly(n)}$ of what ever the true valid value in $[c_1,c_2]$ is, so fixing $M(i,j) = c^*$ can only increase the optimum of the left hand side of \eqref{eq:mainConstraint} by $1/2^{\poly(n)}$ (this follows since all coefficients of the system are bounded by $2^{\poly(n)}$, and additionally, by our magnitude constraints, in any feasible solution, all variables are bounded by $2^{\poly(n)}$). 
	
	After fixing $M(i,j)$ we move on to the next variable. Note that after adding the constraint $M(i,j) = c^*$ our system still has coefficients bounded in magnitude by $2^{\poly(n)}$ and we can still solve it in $2^{O(k^2)} \cdot \poly(n)$ time. We perform $\poly(n)$ such binary  searches, each requiring $\poly(n)$ polynomial system verifications, and thus requiring total runtime $\poly(n) \cdot 2^{O(k^2)}$, as claimed in Theorem \ref{thm:lrpd1}.
	
	\subsubsection{Extension to Factor Analysis}
	
	The same general technique described above can be used to solve the polynomial system that arises in the factor analysis algorithm of Theorem \ref{thm:fa1}. The only difference is that we add $O(k)$ additional constraints to our original system, requiring that $D_{TT}(i,i) \ge 0$ for all $i$. Each system solve thus takes $\poly(n)\cdot k^{O(k^2)} = \poly(n) \cdot 2^{O(k^2\log k)}$ time.
	
	\subsubsection{Use of Polynomial Solvers in Theorems \ref{thm:lrpd2} and \ref{thm:fa2}}
	
	Our use of polynomials solvers in Theorems \ref{thm:lrpd2} and \ref{thm:fa2} is very similar to in Theorems \ref{thm:lrpd1} and \ref{thm:fa1}, so we just discuss relevant modifications. Focusing on Theorem \ref{thm:lrpd2}, we reduce solving the LRPD approximation problem to finding matrices $R_U,R_V \in \R^{O(k/\epsilon^2) \times k}$, $Z_T,W_T \in \R^{O(k/\epsilon^2) \times k}$, and $D_{TT} \in \R^{O(k/\epsilon^2) \times O(k/\epsilon^2)}$ minimizing the lefthand side of \eqref{eq:insteadOpt}. 

	As in the exact case, the function to be minimized can be written as a degree four polynomial in $O(k^2/\epsilon^2)$ variables: the entries of the free matrices. We will first perform binary search to identify  the minimum value of this polynomial up to $\frac{1}{2^{\poly(n)}}$ error. We will then use binary search, as in Theorems \ref{thm:lrpd1} and \ref{thm:fa1} to find values for the free matrices that achieve within  $\frac{1}{2^{\poly(n)}}$ of this minimum.
	As before, to 
	apply  binary search we must first bound the range that we must search over. 
	
	\subsubsection{Bounding the Range of Binary Search}
	
	In searching for the minimum value of \eqref{eq:insteadOpt}, we know that it is bounded within $[0, \norm{A}_F^2]$ and thus is bounded in magnitude for $2^{\poly{n}}$.
	
	By the assumption of Theorem \ref{thm:lrpd2} that there exist optima for the LRPD problem $D^\star,L^\star$ with entries bounded in magnitude by $2^{\poly(n)}$. As shown in \eqref{eq:barMatch} there is a near optimal solution to the alternative objective function \eqref{eq:min_problem} using a diagonal matrix that closely matches $D^\star$ on all entries, and thus also has entries bounded in magnitude by  $2^{\poly(n)}$. Finally, we can see by  examining the proof of Claim \ref{clm:approxSpan} that there is a near optimal solution to \eqref{eq:insteadOpt} with the same diagonal entries. Along with the assumption that $A$ has bounded entries and Lemma \ref{lem:factorBound} this implies that there is a solution to \eqref{eq:insteadOpt} where the entries of all unknown matrices are bounded by $2^{\poly(n)}$.
	 
	 \subsubsection{Performing an Iteration of Binary  Search}
	 
	 To perform binary search for the minimum value of \eqref{eq:insteadOpt} we will consider the polynomial system that restricts the lefthand side $\ge c_1$ and $\le c_2$ for $c_1,c_2 \in \left [0,2^{\poly(n)}\right ]$. We will also require that all variables are bounded in magnitude by $2^{\poly(n)}$ (i.e., include a degree-$2$ constraint that bounds the sum of squares of the variables by $2^{\poly(n)}$). Overall this system has three constraints, $O(k^2/\epsilon^2)$ variables, and degree $O(1)$. Its bit complexity can be bounded in the same way as argued for Theorem \ref{thm:lrpd1} and thus it can be verified in $\poly(n) \cdot 2^{O(k^2/\epsilon^2)}$ time. In $\poly(n)$ iterations we can identify the optimum value up to $1/2^{\poly(n)}$ error.
	 
	 To perform binary search on the unknown matrix entries in \eqref{eq:insteadOpt} we will consider the polynomial system consisting of an inequality restricting that the lefthand side is upper bounded by the approximate minimum that we identify, along with the same magnitude constraints.
	 In each iteration of binary search, we will consider an entry $M(i,j)$ (where $M$ is $R_V,R_U,Z_T,W_T$, or $D_{TT}$) and verify this polynomial system augmented with the search constraints $c_1 \le M(i,j)$ and $M(i,j) \le c_2$ for any $c_1,c_2 \in [-2^{\poly(n)},2^{\poly(n)}]$.
	The remainder of the argument exactly mirrors that used for Theorems \ref{thm:lrpd1} and \ref{thm:fa1}. 

\end{document}